\documentclass[a4paper,UKenglish, autoref,usenames,dvipsnames]{lipics-v2019}

\usepackage{amsthm,url,graphics,amsmath,amsfonts,amssymb,verbatim,tikz, capt-of, color, soul,todonotes}
\usepackage[ruled, linesnumbered, vlined]{algorithm2e}
\usepackage{booktabs}
\usepackage{microtype,multirow,mathtools}
\usepackage{paralist}
\usepackage{xcolor}

\def\versiondense#1{#1}
\def\versionspacy#1{}

\versionspacy{
\usepackage{fullpage}
}

\usetikzlibrary{calc, shapes}
\usetikzlibrary{patterns}
\usetikzlibrary{decorations.pathmorphing}
\newcommand{\zigzag}{zigzag}

\newcommand{\tw}{\operatorname{tw}}
\newcommand{\pw}{\operatorname{pw}}
\newcommand{\Oh}{O}

\DeclareMathOperator*{\argmax}{arg\,max}
\makeatletter
\newcommand*\bigcdot{\mathpalette\bigcdot@{.5}}
\newcommand*\bigcdot@[2]{\mathbin{\vcenter{\hbox{\scalebox{#2}{$\m@th#1\bullet$}}}}}
\makeatother

\def\paragraphbf#1{\par\vskip 7pt\noindent{\bf #1}\hskip 10pt}
\def\paragraphsf#1{\par\vskip 7pt\noindent{\sf #1}\hskip 10pt}
\def\inline#1:{\par\vskip 7pt\noindent{\bf #1:}\hskip 10pt}

\newcommand{\tildeG}{\widetilde{\G}}
\newcommand{\C}{\mathcal{C}}
\newcommand{\G}{\mathcal{G}}

\newcommand{\ksum}{$k$-SUM}
\newcommand{\kspm}{$k$-SPM}
\newcommand{\spm}{\textsc{Subset $\Sigma-\Pi$ Maximization}}
\newcommand{\pbds}{\textsc{PBDS}}

\newcommand{\tx}{\tilde{x}}
\newcommand{\ty}{\tilde{y}}

\newcommand{\nex}{\operatorname{e}}
\newcommand{\A}{\mathcal{A}}
\newcommand{\D}{\mathcal{D}}

\newcommand{\wone}{$\mathsf{W}[1]$}
\newcommand{\wtwo}{$\mathsf{W}[2]$}

\newcommand{\bli}{\textbf{i}}
\newcommand{\blj}{\textbf{j}}
\newcommand{\blh}{\textbf{h}}
\newcommand{\blr}{\textbf{r}}
\newcommand{\ca}{\alpha}
\newcommand{\cb}{\beta}
\newcommand{\cg}{\gamma}

\newcommand{\inv}{^{-1}}

\newcommand{\soli}{\texttt{Sol}_\bli}
\newcommand{\solj}{\texttt{Sol}_\blj}
\newcommand{\solh}{\texttt{Sol}_\blh}

\newcommand{\vali}{\texttt{Val}_\bli}
\newcommand{\valj}{\texttt{Val}_\blj}
\newcommand{\valh}{\texttt{Val}_\blh}

\renewcommand{\S}{\mathcal{S}}
\def\Si{\mathcal{S}_\bli}

\newcommand{\T}{\mathcal{T}}
\newcommand{\TD}{\mathtt{T}}
\renewcommand{\Xi}{X_\bli}
\newcommand{\Xj}{X_\blj}
\newcommand{\Xh}{X_\blh}
\newcommand{\Xr}{X_\blr}

\def\dnsitem{\vspace{-5pt}\item}

\newcommand{\Gi}{\G_\bli}
\newcommand{\Ci}{\C_\bli}

\newcommand{\Cj}{\C_\blj}

\newcommand{\Ch}{\C_\blh}
\def\I{{\cal I}}
\def\L{{\cal L}}
\def\X{{\cal X}}
\def\Y{{\cal Y}}

\def\parent{\textsc{par}}
\def\opt{{opt}}
\def\wt{\omega}

\newcommand{\mcc}{\textsc{Multi-Colored Clique} }

\theoremstyle{plain}

\versionspacy{
\newtheorem{theorem}{Theorem}[section]
\newtheorem{lemma}[theorem]{Lemma}
\newtheorem{claim}[theorem]{Claim}
\newtheorem{proposition}[theorem]{Proposition}

}
\newtheorem{observation}[theorem]{Observation}
\newtheorem{conjecture}[theorem]{Conjecture}

\versiondense{
\author{Keerti Choudhary}{Indian Institute of Technology Delhi, India}{keerti@iitd.ac.in}{}{}
\author{Avi Cohen}{Tel Aviv University, Israel}{avicohen2@mail.tau.ac.il}{}{}
\author{N. S. Narayanaswamy}{Department of Computer Science and Engineering, IIT Madras, India}{swamy@cse.iitm.ac.in}{}{}
\author{David Peleg}{Weizmann Institute of Science, Rehovot, Israel}{david.peleg@weizmann.ac.il}{}{}
\author{R.~Vijayaragunathan}{Department of Computer Science and Engineering, IIT Madras, India}{vijayr@cse.iitm.ac.in}{}{}
\titlerunning{Budgeted Dominating Sets in Uncertain Graphs}
\authorrunning{Choudhary et al.}
\date{}
\hideLIPIcs   
\nolinenumbers 
}

\begin{document}

\title{Budgeted Dominating Sets in Uncertain Graphs}

\versionspacy{
\author{Keerti Choudhary
\thanks{Indian Institute of Technology, Delhi, India.
  keerti@iitd.ac.in}
\and
Avi Cohen
\thanks{Tel Aviv University, Israel.
  avicohen2@mail.tau.ac.il}
\and
N. S. Narayanaswamy
\thanks{
  Indian Institute of Technology, Madras, India.
  \{swamy,vijayr\}@cse.iitm.ac.in}{}{}
\and
David Peleg
\thanks{Weizmann Institute of Science, Rehovot, Israel.
  david.peleg@weizmann.ac.il}
\and
R.~Vijayaragunathan~$^\ddag$
}
\date{\today}
}

\maketitle

\begin{abstract}
We study the {\em Budgeted Dominating Set} (BDS) problem on uncertain graphs, namely, graphs with a probability distribution $p$ associated with the edges, such that an edge $e$ exists in the graph with probability $p(e)$.
The input to the problem consists of a vertex-weighted uncertain graph $\G=(V, E, p, \omega)$ and an integer {\em budget} (or {\em solution size}) $k$, and the objective is to compute a vertex set $S$ of size $k$ that maximizes the expected total domination (or total weight) of vertices in the closed neighborhood of $S$. 
We refer to the problem as the {\em Probabilistic Budgeted Dominating Set}~(PBDS) problem. 
In this article, we present the following results on the complexity of the PBDS problem.

\begin{enumerate}
\dnsitem  We show that the PBDS problem is NP-complete even when restricted to uncertain {\em trees} of diameter at most four. 
This is in sharp contrast with the well-known fact that the BDS problem is solvable in polynomial time in trees.
We further show that PBDS is \wone-hard for the budget parameter $k$, and under the {\em Exponential time hypothesis} it cannot be solved in $n^{o(k)}$ time.

\item We show that if one is willing to settle for $(1-\epsilon)$ approximation, then there exists a PTAS for PBDS on trees. 
Moreover, for the scenario of uniform edge-probabilities, the problem can be solved optimally in polynomial time.

\item We consider the parameterized complexity of the PBDS problem, and show that Uni-PBDS (where all edge probabilities are identical) is \wone-hard for the parameter pathwidth. On the other hand, we show that it is FPT in the combined parameters of the budget $k$ and the treewidth.   

\item Finally, we extend some of our parameterized results to planar and apex-minor-free graphs.
\end{enumerate}

Our first hardness proof (Thm. \ref{thm:treehardness}) makes use of
the new problem of $k$-\spm\ ($k$-SPM), which we believe is of independent interest.
We prove its NP-hardness by a reduction from the well-known {\ksum} problem,
presenting a close relationship between the two problems.
\end{abstract}

\section{Introduction}
\label{section:introduction}

\vspace{-10pt}
\paragraphbf{Background and Motivation.}
Many optimization problems in network theory deal with placing resources
in key vertices in the network so as to maximize coverage.
Some practical contexts where such coverage problems occur include
placing mobile towers in wireless networks to maximize reception, assigning
emergency vehicle centres in a populated area to guarantee fast response,
opening production plants to ensure short distribution lines, and so on.
In the context of social networks, the problem of spreading influencers
so as to affect as many of the network members as possible
has recently attracted considerable interest. 

Coverage problems may assume different forms depending on
the optimized parameter.
A basic ``full coverage'' variant is the classical {\em dominating set} problem,
which asks to find a {\em minimal} vertex set $S$ such that each vertex not in
$S$ is {\em dominated by} $S$, i.e., is adjacent to at least one vertex in $S$. 
In the dual {\em budgeted dominating set (BDS)} problem, given a bound $k$
(the budget), it is required to find a set $S$ of size at most $k$
{\em maximizing} the number of covered vertices. 
Over {\em vertex weighted} graphs, the goal is to maximize the total
{\em weight} of the covered vertices, also known as the {\em domination}.
It is this variant that we're concerned with here. 

Traditionally, coverage problems involve a fixed network of static topology.
The picture becomes more interesting when the network structure is uncertain,
due to potential edge connections and disconnections or link failurs.
Pre-selection of resource locations at the design stage becomes more
challenging in such partial-information settings.
  
In this work, we study the problem in one of the most fundamental settings,
where the input is a graph whose edges fail independently with a given
probability. The goal is to find a $k$-element set that maximizes
the {\em expected} (1-hop) coverage (or domination).
Our results reveal that the probabilistic versions of the coverage problem
are significantly harder than their deterministic counterparts, 
and analyzing them require more elaborate techniques.

An {\em uncertain graph} $\G$ is a triple $(V, E, p)$, where $V$ is a set of $n$
vertices, $E \subseteq V \times V$ is a set of $m$ edges, and the function
$p:E \to [0,1]$ assigns a probability of existence to each edge in $E$.
So an $m$ edge uncertain graph $\G$ represents a probability space
consisting of $2^m$ graphs, sometimes called {\em possible worlds},
derived by sampling each edge $e \in E$ independently with probability $p(e)$.
For 
$H = (V, E' \subseteq E)$, the event of sampling $H$ as a possible world,
denoted $H \sqsubseteq \G$, occurs with probability
\versiondense{
$\Pr(H \sqsubseteq \G) =  \prod_{e \in E'} p(e) \prod_{e \in E\setminus E'} \big(1-p(e)\big)$.
}
\versionspacy{
$$\Pr(H \sqsubseteq \G) ~=~  \prod_{e \in E'} p(e) \prod_{e \in E\setminus E'}
\big(1-p(e)\big).$$
}
The notion of possible worlds dates back to Leibniz and
\emph{possible world semantics (PWS)} is well-studied
in the modal logic literature, beginning with the work of Kripke.

Our work focuses on budgeted dominating sets on
vertex-weighted uncertain graphs,
i.e., the {\em Probabilistic Budgeted Dominating Set} (PBDS) problem.
The input consists of a vertex-weighted uncertain graph $\G=(V,E,p,\wt)$,
with a weight function $\wt:V \to \mathbb{Q}^+$ and an integer budget $k$.
Set $p(vv)=1$ for every $v$. For a vertex $u$ and a set $S \subseteq V$,
denote by
\versiondense{
$\Pr(u \sim S) = 1-\prod_{v \in S}(1-p(uv))$
}
\versionspacy{
$$\Pr(u \sim S) ~=~ 1-\prod_{v \in S}(1-p(uv))$$
}
the probability that 
$u\in S$ or $u$ is connected to some vertex in $S$.
For sets $S_1, S_2 \subseteq V$, the expected coverage (or domination)
of $S_1$ by $S_2$ is defined as
$\C(S_1,S_2) = \sum_{v \in S_1}\big(w(v) \Pr(v \sim S_2)\big).$
The PBDS problem aims to find a set $S$ of size $k$ that maximizes $\C(V,S)$
over the possible worlds.
Its decision version is defined as follows.

\smallskip
\noindent
\fcolorbox{gray!20}{gray!18}{
\parbox{13.5cm}{
{\bf  Probabilistic budgeted dominating set (PBDS)} \\[1mm]
\textsf{\bfseries Input:} A vertex-weighted uncertain graph $\G=(V, E, p, \wt)$,
an integer $k$ and a target domination value $t$.\\
\textsf{\bfseries Question:} Is there a set $S \subseteq V$ of size at most $k$ such that $\C(V, S) \geq t$ ?
}
}
\vspace{.25cm}

\paragraphbf{Our Results and Discussion.}
The budgeted dominating set problem is known to have a polynomial time solution
on trees. A natural question is if the same applies to the probabilistic
version of the problem.
We answer this question negatively, showing the following.

\begin{theorem}
\label{thm:treehardness}
The {\sc PBDS} problem is NP-hard on uncertain trees of diameter $4$. Furthermore, 
(i)~the {\sc PBDS} problem on uncertain trees is \wone-hard for the parameter $k$, and
(ii)~an $n^{o(k)}$ time solution to {\sc PBDS} will falsify the Exponential time hypothesis.
\end{theorem}

In order to prove the theorem, we introduce the following problem.

\smallskip
\noindent
\fcolorbox{gray!20}{gray!18}{
\parbox{13.5cm}{
\spm~{($k$-SPM)} \\[1mm]
\textsf{\bfseries Input:} A multiset ${\cal A} = \{ (x_1,y_1), \ldots, (x_N,y_N)\}$ of $N$ pairs of positive rationals, 
an integer $k$, and a rational $t$.\\[1mm]
\textsf{\bfseries Question:} Is there a set $S \subseteq [N]$ of size exactly $k$ satisfying
$\sum_{i \in S} x_i - \prod_{i \in S} y_i \geq t$ ?
}
}
\smallskip

To establish the complexity of the $k$-SPM problem, we 
present a polynomial time reduction from $k$-SUM to $k$-SPM,
thereby proving that both $k$-PBDS and $k$-SPM are NP-hard.
Moreover, Downey and Fellows \cite{DowneyFellows92} showed
that the $k$-SUM problem is \wone-hard, implying that if $k$-SUM has an FPT
solution with parameter $k$, then the $\mathsf{W}$ hierarchy collapses.
This provides our second hardness result.

\begin{theorem}
\label{thm:kspmwonehard}
The $k$-{\sc SPM} problem is \wone-hard for the parameter $k$. 
Furthermore, any $N^{o(k)}$ time solution to $k$-{\sc SPM} falsify the Exponential time hypothesis.
\end{theorem}

The $k$-SUM problem can be solved easily in $\widetilde O(n^{\lceil k/2\rceil})$
time. However, it has been a long-standing open problem to obtain any polynomial
improvement over this bound~\cite{AbboudL13, Patrascu10}.
Patrascu and Williams~\cite{PatrascuW10} showed an $n^{o(k)}$ time algorithm
for $k$-SUM falsifies the famous {\em Exponential time hypothesis} (ETH).
Hence, our polynomial time reductions also imply that any algorithm optimally
solving $k$-PBDS or $k$-SPM must require $n^{\Omega(k)}$ time unless ETH fails.

\begin{theorem}
\label{thm:pbdslowerbound}
Under the $k$-{\sc SUM} conjecture, for any $\varepsilon>0$, there does not exist  an $n^{\lceil k/2\rceil - \varepsilon}$ 
time algorithm to {\sc PBDS} problem on vertex-weighted uncertain trees. 
\end{theorem}

An intriguing question is whether the $k$-SPM is substantially harder than
$k$-SUM. For the simple scenario of $k=2$, the $2$-SUM problem has
an $O(n \log n)$ time solution.
However, it is not immediately clear whether the $2$-SPM problem has a truly
sub-quadratic time solution (i.e., $O(n^{2-\varepsilon})$ time for some
$\varepsilon>0$). We leave this as an open question.
This is especially of interest due to the following result.

\begin{theorem}
\label{thm:twospmforpbds}
Let $1 \leqslant c < 2$ be the smallest real such that
$2$-{\sc SPM} problem has an $\widetilde O(n^{c})$ time algorithm. 
Then, there exists an $\widetilde O\big((d n)^{c\lceil k/2\rceil +1}\big)$ time
algorithm for optimally solving $k$-{\sc PBDS}  
on trees with arbitrary edge-probabilities, for some constant $d>0$.
\end{theorem}

Given the hardness of $k$-PBDS on uncertain trees, it is of interest to
develop efficient approximation algorithms.  
Clearly, the expected neighborhood size of a vertex set is a submodular
function, and thus it is known that the greedy algorithm yields a
$(1-1/\nex)$-approximation
for the PBDS problem in general uncertain graphs~\cite{KempeKT03,NemhauserWF78}.
For uncertain {\em trees}, we improve this
by presenting a fully polynomial-time approximation scheme for PBDS.

\begin{theorem}
\label{thm:FPTAS}
For any integer $k$, and any $n$-vertex tree with arbitrary edge probabilities, 
a $(1-\epsilon)$-approximate solution to the optimal probabilistic budgeted dominating set 
$(${\sc PBDS}$)$ of size $k$ can be computed in time $\widetilde O(k^2\epsilon^{-1}n^2)$.
\end{theorem}

We also consider a special case that the number of distinct probability edges on the input uncertain tree is bounded above by some constant $\gamma$. 
\begin{theorem}
\label{thm:boundedProb}
For any integer $k$, and an $n$-vertex tree $\T$ with at most $\gamma$ edge  probabilities,  an optimal solution for the {\sc PBDS} problem on $\T$ can be computed in time $\widetilde O(k^{(\gamma+2)}n)$.
\end{theorem}

We investigate the complexity of PBDS on bounded treewidth graphs.
The hardness construction on bounded treewidth graphs is much more challenging.
Due to this inherent difficulty,
we focus on the {\em uniform} scenario, where all edge probabilities $p(e)$
are identical.
\versionspacy{
(Such a scenario may arise naturally, e.g., in settings where all links are
provided by the same manufacturer and operate in a similar environment.)}
We refer to this version of the problem as \emph{Uni-PBDS}. 
We show that for any $0 < q < 1$, the Uni-PBDS
problem with edge-probability $q$
is \wone-hard for the pathwidth parameter of the input uncertain graph $\G$.
In contrast, the BDS problem (when all probabilities are one) is
\versiondense{
FPT
}
\versionspacy{
fixed-parameter tractable
}%
when parameterized by the pathwidth
of the input graph.

\vspace{-3pt}
\begin{theorem}
\label{thm:whardpathwidth}
Uni-PBDS is \wone-hard w.r.t. the pathwidth of the input uncertain graph. 
\end{theorem}

Then, we consider the Uni-PBDS problem with combined $k$ and treewidth
parameters. 
We show that the Uni-PBDS problem can be formulated as a 
variant of the Extended Monadic Second order  (EMS)  problem due to Arnborg et al.~\cite{ArnborgLS91}, to derive an
FPT algorithm for the Uni-PBDS problem
parameterized by the treewidth of $\G$ and $k$.  
\begin{theorem}
\label{thm:twkfpt}
For any integer $k$, and any $n$-vertex uncertain graph of treewidth $w$
with uniform edge probabilities, 
$k$-Uni-{\sc PBDS}
can be solved in time $O(f(k,w) n^2)$, and thus is FPT in the combined parameter involving $k$ and $w$.
Furthermore, $f(k,w)$ is $k^{O(w)}$. 
\end{theorem}

Finally, using the structural property of dominating sets from
Fomin et al.~\cite{FominLRS11}, we derive FPT algorithms
parameterized by the budget $k$ in apex-minor-free graphs and planar graphs. 

\vspace{-3pt}
\begin{theorem}
For any integer $k$, and any $n$-vertex weighted planar or apex-minor free
graph, the Uni-PBDS problem can be solved in time $2^{O(\sqrt{k} \log k)}n^{O(1)}$.
\end{theorem}

\paragraphbf{Related Work.}
\versiondense{
Uncertain graphs have been used in the literature to model the uncertainty
among relationships in protein-protein interaction networks in bioinformatics~
\cite{AsthanaKGR04}, road networks~\cite{AnezBP96,HuaPei10} and social
networks~\cite{DomingosRichardson01,KempeKT03,SwamynathanWBAZ08,WhiteHarary01}.
Connectivity~\cite{Ball80,BallP83,GuoJerrum19,KarpL85,ProvanB83,Valiant79},
network flows~\cite{Evans76,FrankHakimi65}, structural-context
similarity~\cite{ZouL13}, minimum spanning trees~\cite{ErlebachHKMR08},
coverage~\cite{Daskin83,HassinRS09,HassinRS17,NarayanaswamyNV18,MelachrinoudisHelander96}, and community detection~\cite{BonchiGKV14,PengZZLQ18} are
well-studied problems on uncertain graphs.
In particular, budgeted coverage problems model a wide variety of interesting
combinatorial optimization problems on uncertain graphs. For example,
the classical facility location problem~\cite{Hochbaum96,KhullerMN99}
is a variant of coverage.
As another example, in a classical work, Kempe, Kleinberg, and Tardos~
\cite{KempeKT03} study influence maximization problem as an expected coverage
maximization problem in uncertain graphs. They consider the scenario where
influence propagates probabilistically along relationships, under different
influence propagation models, like the {\em Independent Cascade (IC)} and
{\em Linear Threshold (LT)} models, and show that choosing $k$ influencers
to maximize the expected influence is NP-hard in the IC model.
The coverage problem in the presence of uncertainty was studied extensively also
in sensor placement and w.r.t. the placement of light sources
in computer vision.
A special case of the budgeted coverage problem is the
{\em Most Reliable Source (MRS)} problem, where given an uncertain graph
$\G = (V, E, p)$, the goal is to find a vertex $u \in V$ such that the expected
number of vertices in $u$'s connected component is maximized.
To the best of our knowledge, the computational complexity of MRS is not known,
but it is polynomial time solvable on some specific graph classes like trees
and series-parallel graphs~\cite{ColbournXue98,Ding09,Ding11,DingXue11,MelachrinoudisHelander96}. 
Domination is another special kind of coverage and its complexity is very
well-studied.  The classical dominating set (DS) problem is known to be \wtwo-hard
in general graphs~\cite{DowneyFellows92},
and on planar graphs it is fixed parameter tractable with respect to the size
of the dominating-set as the parameter~\cite{FrickGrohe99}.
Further, on $H$-minor-free graphs,
the dominating-set problem is solvable in subexponential time~
\cite{AlberFN04,DemaineFHT05jacm}. It also admits a linear kernel on
$H$-minor-free graphs and graphs of bounded expansion~
\cite{AFN04,DrangeDFKLPPRVS16,F.V.Fomin:2010oq,FominLST18,PhilipRS09}.  
On graphs of treewidth bounded by $w$, the classical dynamic programming
approach~\cite{cygan2015parameterized} 
can be applied to show that the DS problem is FPT when parameterized~by~$w$.
}
\versiondense{
The {\em Budgeted Dominating Set (BDS)} problem
is known to be NP-hard~\cite{KneisMR07} as well as \wone-hard for the
budget parameter~\cite{DowneyFellows92}.
Furthermore, a subexponential parameterized algorithm is known~
for BDS on apex-minor-free graphs~\cite{FominLRS11}.
The treewidth-parameterized FPT algorithm for the dominating-set problem can be
adapted to solve the BDS problem in time $\Oh(3^{w}kn)$.
In particular, for trees there exists a linear running time algorithm.
}
\versiondense{
PBDS was studied as {\sc Max-Exp-Cover-1-RF} in the survey paper
\cite{NarayanaswamyV20}, and given a dynamic programming algorithm on
a nice tree decomposition with runtime $2^{\Oh(w \cdot \Delta)}n^{\Oh(1)}$,
where $\Delta$ is the maximum degree of $\G$.
The question whether PBDS has a treewidth parameterized FPT algorithm
remained unresolved; it is settled in the negative in this work. 
}

\section{Preliminaries}
Consider a simple undirected graph $G=(V,E)$ with vertex set $V$ and edge set
$E$, and let $n=|V|$ and $m=|E|$. Given a vertex subset $S \subseteq V$, the 
subgraph induced by $S$ is denoted by $G[S]$. For a vertex $v \in V$,  $N(v)$ denotes 
the set of neighbors of $v$ and $N[v] = N(v) \cup \{v\}$ is the closed neighborhood of 
$v$. Let $deg(v)$ denote the degree of the vertex $v$ in $G$. A vertex subset 
$S \subseteq V$ is said to be a {\em dominating set} of $G$ if every vertex 
$u \in V\setminus S$ has a neighbor $v \in S$. For an integer $r > 0$, a vertex subset 
$S \subseteq V$ is said to be an {\em $r$-dominating set} of $G$ if for every vertex 
$u \in V\setminus S$ there exists a vertex $v \in S$ at distance at most $r$ from $u$.
A graph $H$ is said to be an {\em apex} if it can be made planar by the removal
of at most one vertex. A graph $G$ is said to be {\em apex-minor-free} if it does not 
contain as its minor some fixed apex graph $H$. All planar graphs are apex-minor-free 
as they do not contain as minor the apex graphs $K_{3,3}$ and $K_5$.
The notations $\mathbb{R}$, $\mathbb{Q}$ and $\mathbb{N}$ denote, respectively, 
the sets of real, rational, and natural numbers (including 0).
For integers $a \leq b$, define $[a,b]$ to be the set $\{a, a+1, \ldots, b\}$,
and for $b>0$ let $[b] \equiv [1,b]$.

Other than this, we follow standard graph theoretic and parameterized 
complexity terminology~
\cite{cygan2015parameterized,Diestel12,DowneyFellows13}.


\paragraphsf{Numerical Approximation.}
When analyzing our polynomial reductions, we employ numerical analysis
techniques to bound the error in numbers obtained as products of
an exponential and the root of an integer.
We use the following well-known bound on the error in approximating
an exponential function by the sum of the lower degree terms
in the series expansion.

\begin{lemma}{\bf \cite{Apostol69}}
\label{lem:lagrange}
For $z\in [-1,1]$, $e^z$ can be approximated using the Lagrange remainder as
$$ e^z = 1+z+\frac{z^2}{2!}+\frac{z^3}{3!}+\ldots+\frac{z^Q}{Q!}+R_Q(z)$$
where $|R_Q(z)|\leq e/(Q+1)! \leq  1/2^Q$.
\end{lemma}

We use the following lemma for bounding the error in multiplying approximate
values.
\begin{lemma}
\label{lem:prodLowerbound}
For any set $\{d_1,\ldots,d_k\}$ of $k$ reals in the range $[0,1]$,
$$\prod_{i\in[k]}(1-d_i) ~\geq~ 1- \sum_{i\in[k]}d_i.$$
\end{lemma}

\begin{proof}
The proof is by induction on $k$. The base case of $k=1$ trivially holds.
For any two reals $a,b\in [0,1]$, $(1-a)(1-b)\geq 1-(a+b)$.
Applying this result iteratively yields that for any $k\geq 1$, if
$\prod_{i\in[k-1]}(1-d_i) \geq 1- (\sum_{i\in[k-1]}d_i)$, then
$$\prod_{i\in[k]}(1-d_i) ~\geq~
\bigg(1- \big(\sum_{i\in[k-1]}d_i\big)\bigg) (1-d_{k})
~\geq~ 1- \sum_{i\in[k]}d_i.$$
The claim follows.
\end{proof}

\paragraphsf{Tree Decomposition.}
A {\em Tree decomposition} of an undirected graph $G=(V,E)$ is a pair
$(\TD, X)$, where $\TD$ is a tree whose set of nodes is $X = \{ \Xi \subseteq V \mid \bli \in V(\TD)\}$, such that
\begin{enumerate}
\item for each edge $u \in V$, there is an $\bli \in V(\TD)$ such that $u \in \Xi$,
\item for each edge $uv \in E$, there is an $\bli \in V(\TD)$ such that $u,v \in \Xi$, and
\item for each vertex $v \in V$ the set of nodes $\{\bli \mid v \in \Xi \}$ forms a subtree of $\TD$.
\end{enumerate}

The \emph{width} of a tree decomposition $(\TD, X)$ equals $\max_{\bli \in V(\TD)} |\Xi|-1$.
The treewidth of a graph $G$ is the minimum width over all tree decompositions of $G$.

A tree decomposition $(\TD, X)$ is {\em nice}
if $\TD$ is rooted by a node $\blr$ with $\Xr = \emptyset$ and every node in $\TD$ is 
either an \emph{insert node}, \emph{forget node}, \emph{join node} or \emph{leaf node}.
Thereby,
a node $\bli \in V(\TD)$ is an insert node if $\bli$ has exactly one child $\blj$ such that 
$\Xi = \Xj \cup \{v\}$ for some $v \notin \Xj$;
it is a forget node if $\bli$ has exactly one child $\blj$  such that $\Xi = \Xj \setminus \{v\}$ for some $v \in \Xj$;
it is a join node if $\bli$ has exactly two children $\blj$ and $\blh$ such that $\Xi = \Xj = \Xh$; and
it is a leaf node if $\Xi = \emptyset$.
Given a tree decomposition of width $w$, a nice tree decomposition of width $w$ and $\Oh(w\ n)$ 
nodes can be obtained in linear time~\cite{Kloks94}.

A tree decomposition $(\T, \X)$ is said to be a {\em path decomposition} if $\T$ is a path. 
The {\em pathwidth} of a graph $G$ is minimum width over all possible path decompositions of $G$. 
Let $\pw(G)$ and $\tw(G)$ denote the pathwidth and treewidth of the graph $G$, respectively. 
The pathwidth of a graph $G$ is one lesser than the minimum clique number of an interval supergraph $H$ which contains  $G$ as an induced subgraph. It is well-known that the maximal cliques of an interval graph can be linearly ordered so that for each vertex, the maximal cliques containing it occur consecutively in the linear order.  This gives a path decomposition of the interval graph.  
A path decomposition of the graph $G$ is the  path decomposition of the interval supergraph $H$ which contains $G$ as an induced subgraph. In our proofs we start with the path decomposition of an interval graph and then reason about the path decomposition of graphs that are constructed from it.

\section{Hardness Results on Trees}\label{section:np-hardness}

\subsection{\kspm~hardness}

We first show that the $k$-\spm\ ($k$-SPM) problem is NP-hard by a reduction from the \ksum ~problem.
Let $\langle X,k\rangle$ with $X=\{x_1, \ldots, x_N\}$ be an instance of the \ksum\ problem. 
Let $L=1+\max_{i\in [N]} |x_i|$. 

Denote by $\langle \A, k, t\rangle$ an instance of the \kspm\ problem.
Given an instance $\langle X,k\rangle$ of {\ksum}, we compute the array 
$\A(X) = \{(\tx_i, \ty_i) \mid i \in [N]\}$ of the \kspm\ problem as follows.
For $1\leq i\leq N$, set $\tx_i:=(L+x_i)/(kL)$.

Let $Q =  3\log_2 (kL) $.  For $i \in [N]$, define $y_i=\nex^{x_i/(kL)}$,
and let $\ty_i$ be a rational approximation of $y_i$ 
that is computed using Lemma~\ref{lem:lagrange}
such that $0 \leq y_i-\ty_i \leq 1/2^Q$.
The new instance of the \kspm~problem is $\langle \A(X), k, t=0 \rangle$.

Observe that for each $i\in[N]$,
\versiondense{
$\ty_i\geq y_i-1/2^Q \geq \nex^{-1/k} - 1/(kL)^3 \geq 1/2$,
}
\versionspacy{
$$\ty_i ~\geq~ y_i -\frac{1}{2^Q} ~\geq~ \nex^{-1/k} - \frac{1}{(kL)^3}
~\geq~ \frac{1}{2},$$
}
for $k\geq 3$. Thus, the elements of $\A(X)$ are positive rationals.
The next lemma provides a crucial property of any set $S$ of vertices of size $k$.

\begin{lemma}
\label{lem:precision}
Let $\lambda = (2kL)^{-2}$. 
For each $S \subseteq [N]$ of size $k$,
$$0 ~\leq~ \prod_{i\in S} y_i - \prod_{i\in S}\ty_i ~\leq~ \lambda.$$
\end{lemma}

\begin{proof}
Let $\alpha=1/(kL)^3$. 
We have:
\versiondense{
\begin{align*}
\textstyle
\prod_{i\in S} y_i - \prod_{i\in S}\ty_i 
& \textstyle \leq~ \prod_{i\in S} y_i - \prod_{i\in S}(y_i-\alpha)\\
& \textstyle =~ \prod_{i\in S} y_i \Big(1- \prod_{i\in S}\big(1-\frac{\alpha}{y_i}\big) \Big) \\
& \textstyle \leq~ \prod_{i\in S}y_i \big(\sum_{i\in S}\frac{\alpha}{y_i}\big)\\
& \textstyle \leq~ \nex^{\sum_{i\in S}x_i/(kL)} \alpha k \nex^{1/k} 
~ \leq~ \alpha k\nex^2 ~ \leq \frac{1}{4(kL)^2}~,
\end{align*}
}
\versionspacy{
\begin{align*}
\prod_{i\in S} y_i - \prod_{i\in S}\ty_i 
& \leq~ \prod_{i\in S} y_i - \prod_{i\in S}(y_i-\alpha)
~=~ \prod_{i\in S} y_i \Big(1-\prod_{i\in S}\big(1-\frac{\alpha}{y_i}\big) \Big) \\
& \leq~ \prod_{i\in S}y_i \big(\sum_{i\in S}\frac{\alpha}{y_i}\big)
~\leq~ \nex^{\sum_{i\in S}x_i/(kL)} \alpha k \nex^{1/k} 
~ \leq~ \alpha k\nex^2 ~ \leq \frac{1}{4(kL)^2}~,
\end{align*}
}
where the second inequality is obtained by Lemma~\ref{lem:prodLowerbound}.
The claim follows.
\end{proof}

\versionspacy{ 
\APPENDPROOFlemprecision
}

We now establish the correctness of the reduction. 

\begin{theorem}
\label{thm:sspmHard}
The $k$-{\sc SUM} problem is polynomial-time reducible to $k$-{\sc SPM}.
\end{theorem}

\begin{proof}
Let $M=\sum_{i\in[N]} (x_i+L)$. 
Define a real valued function $F(z) = z-\nex^{-1+z}$ 
with domain $[0,M/(kL)]$. Observe that $F(1)=0$ and the derivative is $F'(1)=0$.
The function $F(\cdot)$ is clearly concave, which indicates that:
\begin{enumerate}[(i)]
\item $F(z)\leq 0$, for each $z\in [0,M/(kL)]$,
\item $F(z)$ obtains its unique maximum at $z=1$, where its value is $0$, and
\item When restricted to the values in the
set $\big\{z/(kL) \mid z\in [0,M] \text{ is an integer}\big\}$,
$F(z)$ obtains its second largest value at $z=1-1/(kL)$.
\end{enumerate}
For any $S\subseteq [N]$, denote $z_S = \sum_{i \in S} \tx_i$.
For a set $S\subseteq [N]$ of size $k$, we have:
\versiondense{
\begin{align}
\textstyle \sum_{i \in S} \tx_i - \prod_{i \in S} y_i
&~=~ z_S - \nex^{\sum_{i\in S} x_i/(kL)}
~=~ z_S - \nex^{\sum_{i\in S} \tx_i/(kL) - 1/k}\\
&~=~ z_S - \nex^{-1}\cdot \nex^{z_S/(kL)}
~=~ F(z_S).
\label{eqn:Fz_S}
\end{align}
}
\versionspacy{
\begin{align}
\sum_{i \in S} \tx_i - \prod_{i \in S} y_i
&~=~ z_S - \nex^{\sum_{i\in S} x_i/(kL)}
~=~ z_S - \nex^{\sum_{i\in S} \tx_i/(kL) - 1/k}\\
&~=~ z_S - \nex^{-1}\cdot \nex^{z_S/(kL)} ~=~ F(z_S).
\label{eqn:Fz_S}
\end{align}
}
By combining Lemma~\ref{lem:precision} and Eq.~\eqref{eqn:Fz_S}, we obtain the following.
\versiondense{
$$\textstyle F(z_S) ~\leq~ \sum_{i\in S}\tx_i - \prod_{i\notin S}\ty_i ~\leq~ F(z_S) + \lambda.$$
}
\versionspacy{
$$F(z_S) ~\leq~ \sum_{i\in S}\tx_i-\prod_{i\notin S}\ty_i ~\leq~ F(z_S)+\lambda.$$
}
On the other hand, for any set $S\subseteq [N]$ of size $k$ for which
$F(z_S) < 0$, we have
\versiondense{
$F(z_S) \leq F(1-1/(kL))= (1 -1/(kL)-e^{-1/(kL)})$.
}
\versionspacy{
$$F(z_S) ~\leq~ F(1-1/(kL)) ~=~ (1 -1/(kL)-e^{-1/(kL)}).$$
}
Further, 
\versiondense{
$$\textstyle 1 -\frac{1}{kL}-e^{-1/(kL)} ~\leq~
\big(1 -\frac{1}{kL}\big)-\big(1 -\frac{1}{kL}+\frac{1}{2(kL)^2}-\frac{1}{6(kL)^3}\big) 
~\leq~ -\frac{1}{4(kL)^{2}} ~=~ -\lambda~.$$
}
\versionspacy{
$$1 -\frac{1}{kL}-e^{-1/(kL)} ~\leq~
\big(1 -\frac{1}{kL}\big)-\big(1 -\frac{1}{kL}+\frac{1}{2(kL)^2}-\frac{1}{6(kL)^3}\big) 
~\leq~  -\frac{1}{4(kL)^{2}} ~=~ -\lambda~.$$
}
So, for a set $S$, $\sum_{i \in S} \tx_i - \prod_{i \in S} \ty_i\geq 0$
if and only if $\sum_{i \in S} \tx_i = 1$, or equivalently $\sum_{i \in S} x_i =0$. 
It follows that $\langle X,k\rangle$ is a yes instance of the \ksum ~problem 
if and only if $\langle \A(X), k, t=0 \rangle$ is a yes instance of the \kspm~problem.
The time to compute $\tx_i$ and $\ty_i$ from $x_i$ is polynomial in $Q\cdot\log_2 (kL)$, for $1\leq i\leq N$.
Thus, the time-complexity of our reduction is $N\cdot\log_2^{O(1)}(kL)$, which is at most polynomial in $N$
as long as $L=2^{O(N)}$.
Hence, the \ksum ~problem is polynomial-time reducible to the \kspm~problem. 
\end{proof}

\begin{proof}[Proof of Theorem~\ref{thm:kspmwonehard}]
The reduction given in the proof of Theorem~\ref{thm:sspmHard} is a parameter preserving reduction for the parameter $k$. 
That is, the parameters in the instances of the \ksum ~and the \kspm~
problem are same in values and the constructed instance of the \kspm~problem is of size polynomial in the input size of the \ksum ~instance. Thus, the reduction preserves the parameter $k$. 
Since the \ksum ~problem is known to be \wone-hard for the parameter $k$
~\cite{AbboudLW14,DowneyFellows92}, 
the \kspm~problem is also \wone-hard for the parameter $k$.
Further, it is known that under the Exponential time hypothesis (ETH),
there cannot exist an $o(N^k)$ time solution for the \ksum\ problem
~\cite{PatrascuW10}, so it follows that under ETH there is no $o(N^k)$  
time algorithm for \kspm~as well. 
\end{proof}

\subsection{Hardness of PBDS on Uncertain Trees}

In this subsection, we show the hardness results for the \pbds~problem on trees,
establishing Theorem~\ref{thm:treehardness}.
In order to achieve this, we present a polynomial time reduction from \kspm~to
{\sc PBDS} on unweighted trees. 


\begin{proof}[Proof of Theorem~\ref{thm:treehardness}]
In order to prove our claim,
we provide a reduction from $k$-{\sc SPM} to {\sc PBDS}.
Given an instance $\langle {\cal A}=((x_1,y_1),\ldots,(x_N,y_N)), k, t \rangle$ of the 
$k$-SPM~problem, where 
$t$ is a rational, an equivalent instance of the {\sc PBDS} problem 
is constructed as follows. Let $n=N^2+N+1$. 
Construct an uncertain tree $\T = (V,E,p)$, 
where the vertex set $V$ consists of three disjoint sets, namely, 
$A=\{a_0\}$, $B=\{b_1,\ldots,b_N\}$, and $C=\{c_{11},c_{12},\ldots,c_{NN}\}$. 
(see Figure~\ref{Figure:hardness}).
Note that the uncertain tree $\T$ is considered to be unweighted or unit weight on the vertices.
The vertex $a_0$ is connected by edges to the vertices in $B$. 
For each $1 \leq i \leq n$, the vertex $b_i$ is connected by edges to the vertices $c_{i1}\ldots, c_{iN}$.
Let $X_{\max}=\max\{1,x_1,x_2,\ldots,x_N\}$ and $Y_{\max}=\max\{1,y_1,y_2,\ldots,y_N\}$.
To complete the construction, define the probability function $p:E\to [0,1]$
as follows:
\[p(v\bar v) = \begin{cases}
r_i = 1 - (y_i)/(X_{\max}\cdot Y_{\max}),
	& \text{ if } v\bar v=a_0b_i \text{ for } 1 \leq i \leq N,\\
q_i = x_i/(X_{\max}\cdot Y_{\max})^k,
	& \text{ if } v\bar v=b_ic_{i1} \text{ for } 1 \leq i \leq N,\\
1, & \text{ otherwise.}
\end{cases}\]
Since $x_i,y_i\gneq 0$ for each $1 \leq i \leq n$, we have that
$p(v,\bar v) \in [0,1]$ is rational for every $(v,\bar v)\in E$.
This completes the construction of the instance for the PBDS problem.
We show that the given instance $\langle \A, k, t \rangle$ is a yes instance
of \kspm~if and only if $\T$ has a set $S$ of size $k$ such that
$\C(V, S) \geq 1+(N-1)k + t/(X_{\max}Y_{\max})^k$.

\begin{figure}[h]
\centering
\begin{tikzpicture}[scale=1.1]
\draw[black, thick] (0,0) -- (-2,-1);
\draw[black, thick] (0,0) -- (0,-1) node[midway, blue] {{$r_i~~~$}};
\draw[black, thick] (0,0) -- (2,-1);
\draw[black, thick] (-2,-1) -- (-2.75,-2);
\draw[black, thick] (-2,-1) -- (-2,-2);
\draw[black, thick] (-2,-1) -- (-1.25,-2);
\draw[black, thick] (0,-1) -- (-0.75,-2) node[midway, left, blue] {{$q_i$}};;
\draw[black, thick] (0,-1) -- (0,-2) node[midway, blue] {{$1~~$}};;
\draw[black, thick] (0,-1) -- (0.75,-2) node[midway, blue] {{$~~~1$}};;
\draw[black, thick] (2,-1) -- (2.75,-2);
\draw[black, thick] (2,-1) -- (2,-2);
\draw[black, thick] (2,-1) -- (1.25,-2);
\draw[black,fill=MidnightBlue] (0,0) circle (0.05cm) node[above, black] {{$a_0$}};
\draw[black,fill=MidnightBlue] (-2,-1) circle (0.05cm) node[left, black] {{$b_1$}};
\draw[black,fill=MidnightBlue] (0,-1) circle (0.05cm) node[left, black] {{$b_i$}};
\draw[black,fill=MidnightBlue] (2,-1) circle (0.05cm) node[left, black] {{$b_n$}};
\draw[black,fill=MidnightBlue] (-2.75,-2) circle (0.05cm) node[below, black] {{$c_{11}$}};
\draw[black,fill=MidnightBlue] (-2,-2) circle (0.05cm) node[below, black] {{$c_{12}$}};
\draw[black,fill=MidnightBlue] (-1.25,-2) circle (0.05cm) node[below, black] {{$c_{1n}$}};
\draw[black,fill=MidnightBlue] (-0.75,-2) circle (0.05cm) node[below, black] {{$c_{i1}$}};
\draw[black,fill=MidnightBlue] (0,-2) circle (0.05cm) node[below, black] {{$c_{i2}$}};
\draw[black,fill=MidnightBlue] (0.75,-2) circle (0.05cm) node[below, black] {{$c_{in}$}};
\draw[black,fill=MidnightBlue] (1.25,-2) circle (0.05cm) node[below, black] {{$c_{n1}$}};
\draw[black,fill=MidnightBlue] (2,-2) circle (0.05cm) node[below, black] {{$c_{n2}$}};
\draw[black,fill=MidnightBlue] (2.75,-2) circle (0.05cm) node[below, black] {{$c_{nn}$}};
\node[black, thick] at (-1,-1) {$\ldots$};
\node[black, thick] at (1,-1) {$\ldots$};
\node[black, thick] at (-1.625,-2) {$\ldots$};
\node[black, thick] at (0.375,-2) {$\ldots$};
\node[black, thick] at (2.375,-2) {$\ldots$};
\end{tikzpicture}
\caption{Illustration of the lower bound of Theorem~\ref{thm:treehardness}.
Here $p_i=1 - y_i/(X_{\max}\cdot Y_{\max})$ and $q_i=x_i/(X_{\max}\cdot Y_{\max})^k$
for $i\in [N]$.}
\label{Figure:hardness}
\end{figure}
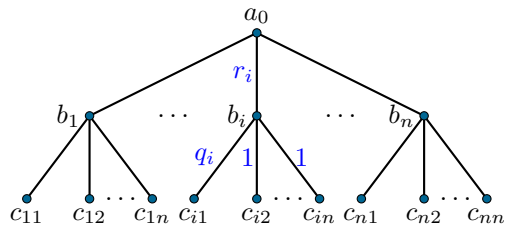

Let $S_{opt}$ be a set of size $k$ maximizing $\C(V, S_{opt})$ in $\T$.
We show $S_{opt}\subseteq B$.
Assume, to the contrary, that there exists some $z\in S_{opt}$ satisfying $z\notin B$.
Consider $i \in [N]$ such that
none of the vertices $c_{i1},\ldots,c_{iN}$ lie in $S_{opt}$.
Such $i$ must exist since $|S_{opt}|=k$.
If $z\in C$, then replacing $z$ by $b_i$ results in a set
$S' = S_{opt} \setminus \{z\} \cup \{b_i\}$ such that
$\C(V,S') \geq \C(V, S_{opt}) + N-3$, contradicting the optimality of $S_{opt}$.
Hence, $S_{opt}$ must be contained in $A\cup B$. 
In this case, $z$ must be $a_0$.  Now the set
$S' =  S_{opt} \setminus \{z\} \cup \{b_i\}$ is such that
$\C(V,S') \geq \C(V, S_{opt}) + N/2-2$. 
Since $N \geq 6$, this contradicts the optimality of $S_{opt}$.
It follows that $S_{opt} \subseteq B$.

Next, consider a set $S\subseteq B$ of size $k$, and let
$I_S=\{i\in [N] \mid b_i\in S\}$. We have
\versiondense{
\begin{align*}
\C(V, S)& \textstyle
=~\left(1-\prod_{i\in I_S}\left(1-p(a_0,b_i)\right)\right) +
\left(\sum_{i\in I_S} \left(p(b_i,c_{i1})+N-1\right)\right)
\\ & \textstyle
=~1+(N-1)k~+ \sum_{i\in I_S} x_i/(X_{\max}\cdot Y_{\max})^{k}~ -
\prod_{i\in I_S} y_i/(X_{\max}\cdot Y_{\max})
\\ & \textstyle
=~ 1+(N-1)k~+ (X_{\max}\cdot Y_{\max})^{-k}
\left(\sum_{i\in I_S} x_i~ -
\prod_{i\in I_S} y_i \right)~.
\end{align*}
}
\versionspacy{
\begin{align*}
\C(V, S)& =~\left(1-\prod_{i\in I_S}\left(1-p(a_0,b_i)\right)\right)+
\left(\sum_{i\in I_S} \left(p(b_i,c_{i1})+N-1\right)\right)
\\ & =~1+(N-1)k~+ \sum_{i\in I_S} x_i/(X_{\max}\cdot Y_{\max})^{k}~ -
\prod_{i\in I_S} y_i/(X_{\max}\cdot Y_{\max})
\\ & =~ 1+(N-1)k~+ (X_{\max}\cdot Y_{\max})^{-k}
\left(\sum_{i\in I_S} x_i~ -
\prod_{i\in I_S} y_i \right)~.
\end{align*}
}
This formulation of the coverage function shows that the given instance
$\langle \A, k, t \rangle$ is a yes instance of \kspm~if and only if $\T$ has
a set $S$ of vertices of size $k$ such that 
$\C(V, S) \geq 1+(N-1)k+ t/(X_{\max}Y_{\max})^k$.
Thus, the \kspm~problem is reduced in polynomial time  to 
{\sc PBDS} on unweighed trees.

It remains to prove 
NP-hardness, \wone-hardness, and $n^{o(k)}$ lower-bound under {\sc ETH}.
Note~that the above reduction is a parameterized preserving reduction
for the parameter $k$. 
That is, the parameter $k$ in the \kspm~problem is the solution size (also called $k$) parameter for 
the PBDS problem. Since the \kspm~problem 
(i) is NP-hard, 
(ii) is \wone-hard for the parameter $k$, and
(iii) cannot have time complexity $n^{o(k)}$ under the Exponential time
hypothesis (by Theorem~\ref{thm:kspmwonehard}),
it follows that the same hardness results hold for {\sc PBDS} as well.
Therefore, the {\sc PBDS} problem on uncertain trees 
(i) is NP-hard, 
(ii) is \wone-hard for the parameter $k$, and 
(iii) cannot have time complexity $n^{o(k)}$ if ETH holds true.
\end{proof}

The $k$-SUM conjecture~\cite{AbboudL13,Patrascu10} states that the \ksum, for the parameters $N$ and $k$,
requires at least $N^{\lceil k/2\rceil - o(1)}$ time.

\begin{conjecture}[$k$-{\sc SUM} Conjecture] 
There do not exist a $k \geq 2$, an $\varepsilon > 0$, and 
an algorithm that succeeds (with high probability) in solving
$k$-{\sc SUM}~in $N^{\lceil k/2\rceil - \varepsilon}$ time.
\end{conjecture}

\begin{proof}[Proof of Theorem~\ref{thm:pbdslowerbound}]
Consider the uncertain tree $\T$ constructed in the proof of Theorem~\ref{thm:treehardness}.
We set $n_0=0$. 
Modify the original construction of $\T$ by deleting the $N-2$ vertices: $c_{i3},c_{i4},\ldots,c_{iN}$,
and setting $\wt_{c_{i2}}=N-1$, for $1\leq i\leq N$.
Thus the tree contains exactly $n=3N+1$ vertices. 
Now, \ksum~is reducible to $k$-SPM, and \kspm\ is reducible to PBDS, both in polynomial time,
and, moreover the parameter $k$ remains unaltered and the size of problem grows by at most constant factor.
This shows that, for $\varepsilon>0$, an $n^{\lceil k/2\rceil - \varepsilon}$ time algorithm
to weighted {\sc PBDS} implies an $N^{\lceil k/2\rceil - \Omega(\varepsilon)}$ time algorithm to \ksum,
thereby, falsifying the $k$-SUM conjecture.
\end{proof}

Note that since the PBDS problem is NP-hard on trees, it is also {\tt para-NP}-hard
\cite{cygan2015parameterized,DowneyFellows13} for the treewidth parameter.  

\section{Hardness of Uni-PBDS for the pathwidth parameter}
\label{sec:whard}
In this section, we show that even for the restricted case of uniform
probabilities, the Uni-PBDS problem is \wone-hard for the pathwidth parameter,
and thus also for treewidth (Theorem~\ref{thm:whardpathwidth}).
This is shown by a reduction from the \mcc problem to the Uni-PBDS problem.
It is well-known that the \mcc problem is \wone-hard for the parameter
solution size~\cite{DowneyFellows92}.

\smallskip\smallskip
\noindent
\fcolorbox{gray!20}{gray!18}{
\parbox{13.5cm}{
\mcc\\
\textsf{Input}: A positive integer $k$ and a $k$-colored graph $G$.\\
{\sf Parameter}: $k$\\
{\sf Question}: Does there exist a clique of size $k$ with one vertex from each color class?
}
}
\smallskip

\noindent
Let $(G=(V,E),k)$ be an input instance of the \mcc problem, with $n$ vertices and $m$ edges. 
Let $V=(V_1,\ldots, V_k)$ denote the partition of the vertex set $V$ in the input instance. 
We assume, without loss of generality, $|V_i| = n$ for each $i \in [k]$. 
For each $1 \leq i \leq k$, let $V_i = \{u_{i,\ell} \mid 1 \leq \ell \leq n\}$. 

\subsection{Gadget based reduction from \mcc}

Let $(G, k)$ be an instance of the \mcc~problem. 
For any probability $0 < p < 1$, and for any integer $f$ such that
$f > \max\{knm, n+k^2/p\}$,
our reduction constructs an uncertain graph $\G$. 
The output of the reduction is an instance $(\G, k',t')$ of the Uni-PBDS problem where each edge has probability $p$, $k' = (n+1)(m+kn)$ and $t' = (kn+m)\big((n+1)fp + n+np+1+2(1-(1-p)^n)\big) + 4{k \choose 2}(1-(1-p)^{n+1})$.  In the presentation below, we show that this choice of $k'$ and $t'$ ensures that there is a set of size $k'$ with expected domination at least $t'$ in $\G$ if and only if $G$ has a multi-colored clique of size $k$. 

We first construct a gadget graph to represent the vertices and edges of the input instance of the \mcc problem. 
We construct two types of gadgets, $\D$ and $\I$ in the reduction,
illustrated in Figure \ref{fig:gadgetfig}.
The gadget $\I$ is the primary gadget and $\D$ is a secondary gadget
used to construct $\I$. 
When we refer to a gadget, we mean the primary gadget $\I$ unless the gadget $\D$ is specified. 
For each vertex and edge in the given graph, our reduction has a corresponding gadget. 
The gadget $\D$ is defined as follows.\\ 

\noindent {\bf Gadget of  type $\D$.} 
Given a pair of vertices $u$ and $v$, the gadget $\D_{u,v}$ consists of vertices $u$, $v$, and $f$ additional vertices. 
The vertices $u$ and $v$ are made adjacent to every other vertex. 
We refer to the vertices $u$ and $v$ as {\em heads}, and remaining vertices of $\D_{u,v}$ as {\em tails}, and $u$ are $v$ are said to be connected by the gadget $\D_{u,v}$. 


\begin{observation}
\label{obs:DPathwidth}
The pathwidth of a gadget of type $\D$ is 2.
\end{observation}

\noindent {\bf Gadget of type $\I$.} We begin the construction of the gadget with $2n$ vertices partitioned into two sets where each partition contains $n$ vertices. 
Let $A = \{a_1, \ldots, a_n\}$ and $C = \{c_1, \ldots, c_n\}$ be this partition. 
For each $i\in [n]$, vertices $a_i$ and $c_i$ are connected by the gadget $\D_{a_i, c_i}$. 
Let $h_a$ and $h_c$ be two additional vertices connected by the gadget $\D_{h_a, h_c}$. 
The vertices in the sets $A$ and $C$ are made adjacent to $h_a$ and $h_c$, respectively. 
This completes the construction of the gadget.  
In the reduction, a gadget of type $\I$ is denoted by the symbol $\I$ along with an appropriate subscript based on whether the gadget is associated with a vertex or an edge.

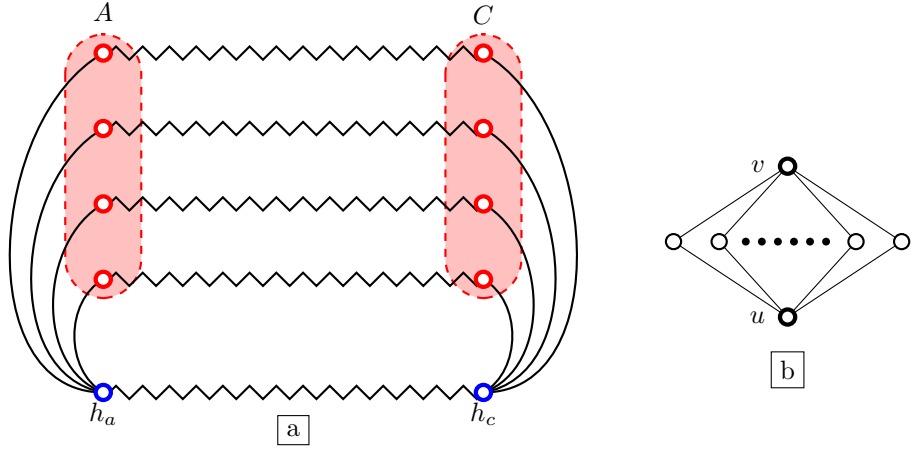
\begin{figure}[h]
\center
\begin{tikzpicture}
\coordinate (a1) at (3,4);
\coordinate (a2) at (3,3);
\coordinate (a3) at (3,2);
\coordinate (a4) at (3,1);

\coordinate (d1) at (8,1);
\coordinate (d2) at (8,2);
\coordinate (d3) at (8,3);
\coordinate (d4) at (8,4);

\coordinate (ha) at (3,-0.5);
\coordinate (hd) at (8,-0.5);

\draw[red, fill=red, fill opacity=0.25, dashed, thick,rounded corners=15pt] ($(a1)+(-0.5,0.25)$) rectangle ($(a4) + (0.5,-0.25)$);
\draw[red, fill=red, fill opacity=0.25, dashed, thick,rounded corners=15pt] ($(d4)+(-0.5,0.25)$) rectangle ($(d1) + (0.5,-0.25)$);

\draw[black, thick] (hd) to [out=30,in=330] (d1);
\draw[black, thick] (hd) to [out=20,in=330] (d2);
\draw[black, thick] (hd) to [out=10,in=330] (d3);
\draw[black, thick] (hd) to [out=0,in=330] (d4);
\draw[black, thick] (ha) to [out=180,in=210] (a1);
\draw[black, thick] (ha) to [out=170,in=210] (a2);
\draw[black, thick] (ha) to [out=160,in=210] (a3);
\draw[black, thick] (ha) to [out=150,in=210] (a4);

\draw[black, thick, decorate, decoration=\zigzag] (d4) -- (a1);
\draw[black, thick, decorate, decoration=\zigzag] (d3) -- (a2);
\draw[black, thick, decorate, decoration=\zigzag] (d2) -- (a3);
\draw[black, thick, decorate, decoration=\zigzag] (d1) -- (a4);
\draw[black, thick, decorate, decoration=\zigzag] (hd) -- (ha);

\draw[red, fill=white, ultra thick] (a1) circle (0.1cm) node[above=0.3cm, black, thick] {$A$};
\draw[red, fill=white, ultra thick] (a2) circle (0.1cm);
\draw[red, fill=white, ultra thick] (a3) circle (0.1cm);
\draw[red, fill=white, ultra thick] (a4) circle (0.1cm);
\draw[red, fill=white, ultra thick] (d1) circle (0.1cm);
\draw[red, fill=white, ultra thick] (d2) circle (0.1cm);
\draw[red, fill=white, ultra thick] (d3) circle (0.1cm);
\draw[red, fill=white, ultra thick] (d4) circle (0.1cm) node[above=0.25cm, black, thick] {$C$};

\draw[blue, fill=white, ultra thick] (ha) circle (0.1cm) node[below, black, thick] {$h_a$};
\draw[blue, fill=white, ultra thick] (hd) circle (0.1cm) node[below, black, thick] {$h_c$};
\node[draw] at (5.5,-1.0) {a};

\makeatletter
\tikzset{
    dot diameter/.store in=\dot@diameter,
    dot diameter=3pt,
    dot spacing/.store in=\dot@spacing,
    dot spacing=10pt,
    dots/.style={
        line width=\dot@diameter,
        line cap=round,
        dash pattern=on 0pt off \dot@spacing
    }
}
\makeatother

\coordinate (a1) at (12,0.5);
\coordinate (a2) at (13.5,1.5);
\coordinate (a3) at (12,2.5);
\coordinate (a4) at (10.5,1.5);

\draw[black] (a1) -- (a2);
\draw[black] (a1) -- (a4);
\draw[black] (a3) -- (a2);
\draw[black] (a3) -- (a4);
\draw[black] (a1) -- ($(a2) + (-0.6,0)$);
\draw[black] (a1) -- ($(a4) + (0.6,0)$);
\draw[black] (a3) -- ($(a2) + (-0.6,0)$);
\draw[black] (a3) -- ($(a4) + (0.6,0)$);

\draw[black, fill=white, ultra thick] (a1) circle (0.1cm) node[left=0.15cm] {$u$};
\draw[black, fill=white, thick] (a2) circle (0.1cm);
\draw[black, fill=white, ultra thick] (a3) circle (0.1cm) node[left=0.15cm] {$v$};
\draw[black, fill=white, thick] (a4) circle (0.1cm);
\draw[black, fill=white, thick] ($(a2) + (-0.6,0)$) circle (0.1cm);
\draw[black, fill=white, thick] ($(a4) + (0.6,0)$) circle (0.1cm);
\node[draw] at (12.0,-0.20) {b};

\draw [black, dot diameter=3pt, dot spacing=6pt, dots] ($(a4) + (0.95, 0)$) -- ($(a2) + (-0.9, 0)$);
\end{tikzpicture}
\caption{
(a) The gadget $\I$ for $n=4$.
(b) The gadget $\D$.
The zigzag edges in $\I$ between two vertices $u$ and $v$ is replaced by
the gadget $\D_{u,v}$. }
\label{fig:gadgetfig}
\end{figure}

\begin{claim}
\label{claim:gagdetPathwidth}
The pathwidth of a gadget of type $\I$ is at most 4. 
\end{claim}

\begin{proof}
We observe that the removal of the vertices $h_a$ and $h_c$ results in a graph in which for each $i \in [n]$,   there is a connected component consisting $a_i$ and $c_i$ which are the heads of a gadget of type $\D$. 
Each component is a gadget of type $\D$ and from Observation \ref{obs:DPathwidth} is of pathwidth 2.
Let $(\T', \X')$ be the path decomposition of $\I-\{h_a,h_c\}$ with width 2. 
Thus adding $h_a$ and $h_c$ into all the bags of the path decomposition $(\T',\X')$ gives a path decomposition for the gadget $\I$, and thus the pathwidth of the gadget $\I$ is at most 4.
\end{proof}

\noindent
{\bf Description of the reduction}.
For $1 \leq i < j \leq k$, let $E_{i,j} = \{xy \mid x \in V_i, y\in V_j\}$
be the set of edges with one end point in $V_i$ and the other in $V_j$ in $G$. 
For each $1 \leq i < j \leq k$, the graph $\G$ has an induced subgraph $\G_i$ corresponding to  $V_i$, and has an induced subgraph $\G_{i,j}$ for the edge set $E_{i,j}$.  
We refer to $\G_i$ as a vertex-partition block and $\G_{i,j}$ as an edge-partition block.  
Inside  block $\G_i$, there is a gadget of type $\I$ for each vertex in $V_i$, and
in the block $\G_{i,j}$, there is a gadget for each edge in $E_{i,j}$.  
For a vertex $u_{i,x}$,  $\I_x$ denotes the gadget corresponding to $u_{i,x}$ in the partition $V_i$, and for an edge $e$, $\I_e$ denotes the gadget corresponding to $e$.  The blocks are appropriately connected by connector vertices which are defined below.

\noindent
We start by defining the structure of a block denoted by $B$. 
The definition of the block applies to both the vertex-partition block and the edge-partition block.
A block $B$ 
consists of gadgets and additional vertices as follows (See Figure \ref{fig:block}).
\begin{itemize}
\item The block $B$ corresponding to the vertex-partition block  $\G_i$ for any $i \in [k]$ is described as follows: for each $\ell \in [n]$, add a gadget $\I_\ell$  to the vertex-partition block $\G_i$, to represent the vertex $u_{i,\ell} \in V_i$. 
In addition to the gadgets, we add $n+1$ vertices to the block $B$ described as follows: Let $F(B) = \{b_1,\ldots,b_n, d_i\}$ be the set of additional vertices that are added to the block $B$. 
For each $\ell \in [n]$, the vertices in the set $C$ of the gadget $\I_{\ell}$ in the block $B$ are made adjacent to $b_\ell$. 
For each $\ell \in [n]$, the vertices in the set $A$ of the gadget $\I_{\ell}$ in the block $B$ are made adjacent to $d_i$. 

\item The block $B$ corresponding to the  edge-partition block $\G_{i,j}$ for any $1 \leq i < j \leq k$ is described as follows: for each $e \in E_{i,j}$, add a gadget $\I_e$ in the edge-partition block $\G_{i,j}$, to  represent the edge $e$. 
In addition to the gadgets, we add $|E_{i,j}|+1$ vertices to the block $B$ described as follows: Let $F(B) = \{b_e \mid e \in E_{i,j}\}\cup\{d_{i,j}\}$ be the set of additional vertices that are added to the block $B$. 
For each $e \in E_{i,j}$, the vertices in the set $C$ of the gadget $\I_e$ in the block $B$ are made adjacent to $b_e$. 
For each $e \in E_{i,j}$, the vertices in the set $A$ of the gadget $\I_e$ in the block $B$ are made adjacent to $d_{i,j}$. 
\end{itemize}

\begin{figure}[h]
\centering
\begin{tikzpicture}
\makeatletter
\tikzset{
    dot diameter/.store in=\dot@diameter,
    dot diameter=3pt,
    dot spacing/.store in=\dot@spacing,
    dot spacing=10pt,
    dots/.style={
        line width=\dot@diameter,
        line cap=round,
        dash pattern=on 0pt off \dot@spacing
    }
}
\makeatother
\coordinate (x1) at (0,0);
\coordinate (r11) at ($(x1) + (0,2)$);
\coordinate (r21) at ($(x1) + (2.5,0.5)$);
\coordinate (a11) at ($(x1) + (0.25,1.75)$);
\coordinate (a21) at ($(x1) + (1.75,1.50)$);
\coordinate (c11) at ($(a11) + (0,-0.75)$);
\coordinate (c21) at ($(a21) + (0,-0.75)$);
\coordinate (x2) at ($(x1) + (3.5,0)$);
\coordinate (r12) at ($(x2) + (0,2)$);
\coordinate (r22) at ($(x2) + (2.5,0.5)$);
\coordinate (a12) at ($(x2) + (0.25,1.75)$);
\coordinate (a22) at ($(x2) + (1.75,1.50)$);
\coordinate (c12) at ($(a12) + (0,-0.75)$);
\coordinate (c22) at ($(a22) + (0,-0.75)$);
\coordinate (x3) at ($(x1) + (10,0)$);
\coordinate (r13) at ($(x3) + (0,2)$);
\coordinate (r23) at ($(x3) + (2.5,0.5)$);
\coordinate (a13) at ($(x3) + (0.25,1.75)$);
\coordinate (a23) at ($(x3) + (1.75,1.50)$);
\coordinate (c13) at ($(a13) + (0,-0.75)$);
\coordinate (c23) at ($(a23) + (0,-0.75)$);
\coordinate (a31) at ($(a21) + (0,0.25)$);
\coordinate (a32) at ($(a22) + (0,0.25)$);
\coordinate (a33) at ($(a23) + (0,0.25)$);
\coordinate (c31) at ($(c11) + (0,-0.25)$);
\coordinate (c32) at ($(c12) + (0,-0.25)$);
\coordinate (c33) at ($(c13) + (0,-0.25)$);
\coordinate (a) at (5.5,2.75);
\coordinate (c1) at ($(r21)+(-1.5,-0.5)$);
\coordinate (c2) at ($(r22)+(-1.5,-0.5)$);
\coordinate (c3) at ($(r23)+(-1.5,-0.5)$);

\draw[black, thick, fill=lightgray, fill opacity=0.3, rounded corners=10pt] (r11) rectangle (r21);
\draw[black, thick, fill=lightgray, fill opacity=0.3, rounded corners=10pt] (r12) rectangle (r22);
\draw[black, thick, fill=lightgray, fill opacity=0.3, rounded corners=10pt] (r13) rectangle (r23);

\filldraw[draw=gray, fill=gray!40] (a11) -- (a31) -- (a) -- (a11) -- cycle;
\filldraw[draw=gray, fill=gray!40] (a12) -- (a32) -- (a) -- (a12) -- cycle;
\filldraw[draw=gray, fill=gray!40] (a13) -- (a33) -- (a) -- (a13) -- cycle;

\filldraw[draw=gray, fill=gray!40] (c21) -- (c31) -- (c1) -- (c21) -- cycle;
\filldraw[draw=gray, fill=gray!40] (c22) -- (c32) -- (c2) -- (c22) -- cycle;
\filldraw[draw=gray, fill=gray!40] (c23) -- (c33) -- (c3) -- (c23) -- cycle;

\draw[black, thick] (a11) -- (a);
\draw[black, thick] (a31) -- (a);
\draw[black, thick] (a12) -- (a);
\draw[black, thick] (a32) -- (a);
\draw[black, thick] (a13) -- (a);
\draw[black, thick] (a33) -- (a);

\draw[black, thick] (c21) -- (c1);
\draw[black, thick] (c31) -- (c1);
\draw[black, thick] (c22) -- (c2);
\draw[black, thick] (c32) -- (c2);
\draw[black, thick] (c23) -- (c3);
\draw[black, thick] (c33) -- (c3);

\draw[red, dashed, thick, fill=red, fill opacity=0.25, rounded corners=2pt] (a11) rectangle (a21);
\draw[red, dashed, thick, fill=red, fill opacity=0.25, rounded corners=2pt] (a12) rectangle (a22);
\draw[red, dashed, thick, fill=red, fill opacity=0.25, rounded corners=2pt] (a13) rectangle (a23);
\draw[red, dashed, thick, fill=red, fill opacity=0.25, rounded corners=2pt] (c11) rectangle (c21);
\draw[red, dashed, thick, fill=red, fill opacity=0.25, rounded corners=2pt] (c12) rectangle (c22);
\draw[red, dashed, thick, fill=red, fill opacity=0.25, rounded corners=2pt] (c13) rectangle (c23);

\draw[blue, thick, fill=white] (a) circle (0.1cm) node[black, thick, above] {$d_i$};
\draw[blue, thick, fill=white] (c1) circle (0.1cm) node[black, thick, below] {$b_1$};
\draw[blue, thick, fill=white] (c2) circle (0.1cm) node[black, thick, below] {$b_2$};
\draw[blue, thick, fill=white] (c3) circle (0.1cm) node[black, thick, below] {$b_n$};
\draw[black, thick, rounded corners=5pt] ($(r11)+(-0.5,1.25)$) rectangle ($(r23)+(0.5,-1.25)$);
\draw [black, dot diameter=4pt, dot spacing=8pt, dots] ($(x2) + (3, 1.25)$) -- ($(x3) + (-.2, 1.25)$);

\node[black, right, ultra thick] at (a21) {$A$};
\node[black, right, ultra thick] at (c21) {$C$};
\node[black, below, ultra thick] at ($(r21) + (-0.25,0)$) {$\I_1$};
\node[black, below, ultra thick] at ($(r22) + (-0.25,0)$) {$\I_2$};
\node[black, below, ultra thick] at ($(r23) + (-0.25,0)$) {$\I_n$};

\end{tikzpicture}
\caption{Illustration of a vertex block $\G_i$ for a $V_i$, $i \in [k]$. Note the $n$ $\I$ gadgets for the $n$ vertices in $V_i$. Similarly, 
an edge block $\G_{i,j}$ for some $1 \leq i < j \leq k$ has $|E_{i,j}|$-many
$\I$ gadgets.
}
\label{fig:block}
\end{figure}
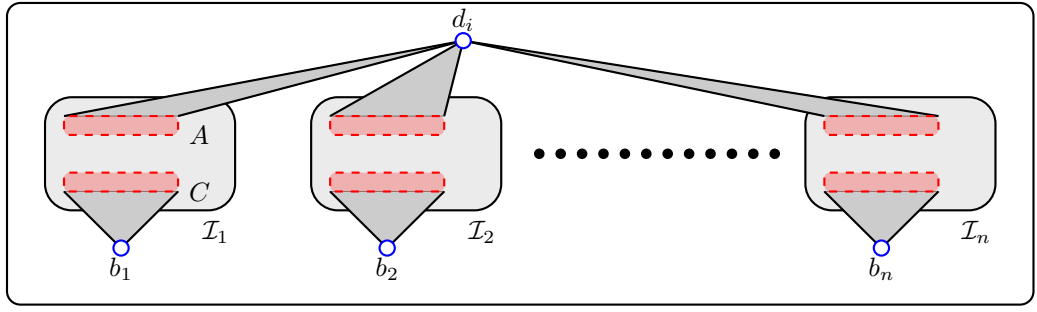

\noindent
The blocks defined above are connected by the connector vertices described next.
These connector vertices are used to connect the edge-partition blocks and vertex-partition blocks, and thus ensure that each edge in $G$ is appropriately represented in $\G$. 
Let $R = \{r_{i,j}^i, s_{i,j}^i, r_{i,j}^j, s_{i,j}^j \mid 1 \leq i < j \leq k\}$ be the connector vertices. 
The blocks are connected based on the cases described below.  The connections involving the $\I$ gadgets in two vertex-partition blocks and an $\I$ gadget in an edge-partition block is illustrated in Figure~\ref{fig:connector}. First, we describe the connection of vertex-partition blocks corresponding $V_i$ and $V_j$ to the appropriate connector vertices. Following this, we describe the connection of the two vertex-partition blocks to the edge-partition block corresponding to $E_{i,j}$ through the appropriate connector vertices.\\
For each $i \in [k]$, each $i < j \leq k$ and each $\ell \in [n]$, 
\begin{itemize}
 \item for each $1 \leq t \leq \ell$, 
 $a_t$ in the gadget 
 $\I_\ell$ of $\G_i$ is made adjacent to 
 $s_{i,j}^i$, and
 \item for each $\ell \leq t \leq n$, 
 $a_t$ in the gadget $\I_\ell$ of $\G_i$ is made adjacent to the vertex $r_{i,j}^i$.
 \end{itemize}
 
For each $i \in [k]$, each $1 \leq j < i$ and each $\ell \in [n]$, 
\begin{itemize}
\item for each $1 \leq t \leq \ell$, 
$a_t$ in the gadget $\I_\ell$ of $\G_i$ is made adjacent to the vertex $s_{j,i}^i$, and
\item for each $\ell \leq t \leq n$, 
$a_t$ in the gadget $\I_\ell$ of $\G_i$ is made adjacent to the vertex $r_{j,i}^i$.
\end{itemize}

\noindent
Now, we describe the edges to connect the $\I$ gadgets in the vertex-partition blocks $\G_i$ and $\G_j$ and to the appropriate $\I$ gadgets in the edge-partition block $\G_{i,j}$. For each $1 \leq i < j \leq k$, and for each $e = u_{i,x}u_{j,y} \in E_{i,j}$, 
\begin{itemize}
\item for each $1 \leq t \leq x$, $a_t$ in the gadget $\I_e$ of $\G_{i,j}$ is made adjacent to the vertex $r_{i,j}^i$, and
\item for each $x \leq t \leq n$, $a_t$ in the gadget $\I_e$ of $\G_{i,j}$ is made adjacent to the vertex $s_{i,j}^i$.
\item for each $1 \leq t \leq y$, $a_t$ in the gadget $\I_e$ of $\G_{i,j}$ is made adjacent to the vertex $r_{i,j}^j$, and
\item for each $y \leq t \leq n$, $a_t$ in the gadget $\I_e$ of $\G_{i,j}$ is made adjacent to the vertex $s_{i,j}^j$.
\end{itemize}

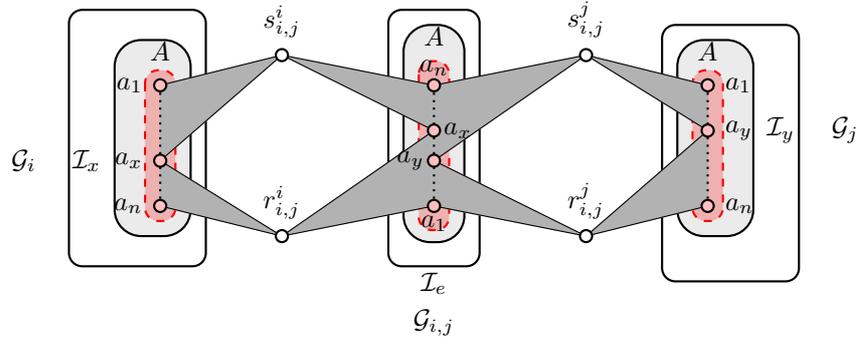
\begin{figure}[h]
\centering
\begin{tikzpicture}[scale=0.8]
\coordinate (a1) at (2,3.5);
\coordinate (a2) at (2,2.25);
\coordinate (a3) at (2,1.5);
\coordinate (b1) at (11,3.5);
\coordinate (b2) at (11,2.75);
\coordinate (b3) at (11,1.5);
\coordinate (c1) at (6.5,3.5);
\coordinate (cx) at (6.5,2.75);
\coordinate (cy) at (6.5,2.25);
\coordinate (c3) at (6.5,1.5);
\coordinate (si) at (4,4);
\coordinate (ri) at (4,1);
\coordinate (sj) at (9,4);
\coordinate (rj) at (9,1);

\draw[thick, black, fill=white, rounded corners=5pt ] ($(a1) + (-1.5, 1.25)$) rectangle ($(a3)+(0.75,-1)$);
\draw[black, thick, fill=lightgray, fill opacity=0.3, rounded corners=10pt] ($(a1)+(-0.75,+0.75)$) rectangle ($(a3)+(0.5,-0.5)$);
\draw[red, fill=red, fill opacity=0.25, dashed, thick,rounded corners=5pt ] ($(a1)+(-0.25,+0.25)$) rectangle ($(a3)+(0.25,-0.25)$);
\filldraw[draw=gray!60, fill=gray!60] (a1) -- (si) -- (a2) -- (a1) -- cycle;
\filldraw[draw=gray!60, fill=gray!60] (a3) -- (ri) -- (a2) -- (a3) -- cycle;
\draw[black] (a1) -- (si);
\draw[black] (a2) -- (si);
\draw[black] (a2) -- (ri);
\draw[black] (a3) -- (ri);
\draw[black, thick, dotted] (a1) -- (a3);
\draw[thick, black, fill=red!25] (a1) circle (0.1cm) node[left=0.1cm] {$a_1$} node[above=0.2cm] {$A$};
\draw[thick, black, fill=red!25] (a2) circle (0.1cm) node[left=0.1cm] {$a_{x}$} node[left=0.65cm] {$\I_x$} node[left=1.5cm] {$\G_i$};
\draw[thick, black, fill=red!25] (a3) circle (0.1cm) node[left=0.1cm] {$a_n$};

\draw[thick, black, fill=white, rounded corners=5pt ] ($(b1) + (-0.75, 1)$) rectangle ($(b3)+(1.5,-1.25)$);
\draw[black, thick, fill=lightgray, fill opacity=0.3, rounded corners=10pt] ($(b1)+(-0.5,+0.75)$) rectangle ($(b3)+(0.75,-0.5)$);
\draw[red, fill=red, fill opacity=0.25, dashed, thick,rounded corners=5pt ] ($(b1)+(-0.25,+0.25)$) rectangle ($(b3)+(0.25,-0.25)$);
\filldraw[draw=gray!60, fill=gray!60] (b1) -- (sj) -- (b2) -- (b1) -- cycle;
\filldraw[draw=gray!60, fill=gray!60] (b3) -- (rj) -- (b2) -- (b3) -- cycle;
\draw[black] (b1) -- (sj);
\draw[black] (b2) -- (sj);
\draw[black] (b2) -- (rj);
\draw[black] (b3) -- (rj);
\draw[black,thick, dotted] (b1) -- (b3);
\draw[thick, black, fill=red!25] (b1) circle (0.1cm) node[right=0.1cm] {$a_1$} node[above=0.2cm] {$A$};
\draw[thick, black, fill=red!25] (b2) circle (0.1cm) node[right=0.1cm] {$a_{y}$} node[right=0.65cm] {$\I_y$} node[right=1.5cm] {$\G_{j}$};
\draw[thick, black, fill=red!25] (b3) circle (0.1cm) node[right=0.1cm] {$a_n$};

\draw[thick, black, fill=white, rounded corners=5pt ] ($(c1) + (-0.75, 1.25)$) rectangle ($(c3)+(0.75,-1)$);
\draw[black, thick, fill=lightgray, fill opacity=0.3, rounded corners=10pt] ($(c1)+(-0.5,1)$) rectangle ($(c3)+(0.5,-0.6)$);
\draw[red, fill=red, fill opacity=0.25, dashed, thick,rounded corners=5pt ] ($(c1)+(-0.25,+0.4)$) rectangle ($(c3)+(0.25,-0.4)$);
\filldraw[draw=gray!60, fill=gray!60] (c1) -- (si) -- (cx) -- (c1) -- cycle;
\filldraw[draw=gray!60, fill=gray!60] (c1) -- (sj) -- (cy) -- (c1) -- cycle;
\filldraw[draw=gray!60, fill=gray!60] (c3) -- (ri) -- (cx) -- (c3) -- cycle;
\filldraw[draw=gray!60, fill=gray!60] (c3) -- (rj) -- (cy) -- (c3) -- cycle;
\draw[black] (c1) -- (si);
\draw[black] (c1) -- (sj);
\draw[black] (cx) -- (ri);
\draw[black] (cx) -- (si);
\draw[black] (cy) -- (rj);
\draw[black] (cy) -- (sj);
\draw[black] (c3) -- (ri);
\draw[black] (c3) -- (rj);
\draw[thick, black, dotted] (c1) -- (c3);
\draw[thick, black, fill=red!25] (c1) circle (0.1cm) node[above] {$a_n$} node[above=0.4cm] {$A$};
\draw[thick, black, fill=red!25] (cx) circle (0.1cm) node[right] {$a_{x}$} node[below=1.75cm] {$\I_e$} node[below=2.25cm] {$\G_{i,j}$};
\draw[thick, black, fill=red!25] (cy) circle (0.1cm) node[left] {$a_{y}$};
\draw[thick, black, fill=red!25] (c3) circle (0.1cm) node[below] {$a_1$};

\draw[thick, black, fill=white] (si) circle (0.1cm) node[above=0.1cm] {$s^i_{i,j}$};
\draw[thick, black, fill=white] (ri) circle (0.1cm) node[above=0.1cm] {$r^i_{i,j}$};
\draw[thick, black, fill=white] (sj) circle (0.1cm) node[above=0.1cm] {$s^j_{i,j}$};
\draw[thick, black, fill=white] (rj) circle (0.1cm) node[above=0.1cm] {$r^j_{i,j}$};
\end{tikzpicture}
\caption{
An illustration of the connector vertices $s_{i,j}^i$, $r_{i,j}^i$, $s_{i,j}^j$
and $r_{i,j}^j$, which connect the blocks $\G_i$ and $\G_{i,j}$,
and $\G_{j}$ and $\G_{i,j}$, for some $1 \leq i < j \leq k$.
The gadget $\I_e$ represents an edge $e = u_{i,x}u_{j,y} \in E_{i,j}$. }
\label{fig:connector}
\end{figure}

This completes the construction of the graph $\G$ with $O(mn^2)$ vertices and $O(mn^3)$ edges. 

\begin{claim}
\label{claim:BlockPathwidth}
The pathwidth of a block $B$ is at most 6. 
\end{claim}

\begin{proof}
Without loss of generality, assume that the block $B$ is a vertex partition
block $\G_i$ for any $i \in [n]$. 
If we remove the vertex $d_i$ from the block $B$, then the resulting graph is a disjoint collection of gadgets of type $\I$ with an additional vertex. 
See Figure~\ref{fig:block} for an illustration. 
By Claim~\ref{claim:gagdetPathwidth}, the pathwidth of a gadget is 4. 
Therefore, for each $\ell \in [n]$, adding the additional vertex $b_\ell$
to all bags of the path decomposition of the gadget $\I_\ell$
gives a path decomposition for the connected component containing the gadget.
Thus, each connected component is of pathwidth at most 5. 
Let $(\T',\X')$ be a path decomposition of $B-\{d_i\}$ with pathwidth 5. 
Thus, adding $d_i$ into all bags of $(\T', \X')$ gives a path decomposition for
the block $B$, and thus the pathwidth of the block is at most 6. 
\end{proof}

The following lemma bounds the pathwidth of the graph $\G$ by a polynomial
in $k$. 

\begin{lemma}
\label{lem:HPathwidth}
The pathwidth of the graph $\G$ is at most $4{k \choose 2} + 6$. 
\end{lemma}

\begin{proof}
Removal of the connector vertices $R$ from $\G$ results in a collection of disjoint blocks. 
By Claim~\ref{claim:BlockPathwidth}, the pathwidth of a block is 6. 
Let $(\T',\X')$ be a path decomposition of $\G-R$ with pathwidth 6. 
Therefore, adding all connector vertices to the path decomposition $(\T', \X')$
gives a path decomposition for the graph $\G$ with pathwidth at most
$4{k \choose 2}+6$.  
\end{proof}

\subsection{Properties of a feasible solution for the Uni-PBDS instance $(\G, k',t')$ output by the reduction}
\label{apdx:twhard:equiv}
\noindent
We start with the observation that in a reduced instance, the expected domination achieved by a set of size $k'$ is at most $t'$.   We then prove properties of a feasible solution for the instance $(\G, k', t')$.
\begin{observation}
\label{obs:maxtprime}
The maximum expected domination that can be achieved by any vertex set of size
$k'$ in $\G$ is $t'$. 
\end{observation}

Let $S$ be a feasible solution for the Uni-PBDS instance $(\G, k', t')$.
We state a set of canonical properties of the set $S$ of size $k'$ and which
achieves the maximum value of $\C(V(\G), S) \geq t'$.  
All the observations below follow crucially from the fact that 
$|S| = k' = (n+1)(kn+m)$ and $\C(V(\G), S) \geq t'
= (kn+m)((n+1)fp + n+np+1+2(1-(1-p)^n)) + 4{k \choose 2}(1-(1-p)^{n+1})$. 

\begin{itemize}
\item Observe that the vertices in the sets $A$, $C$ and $\{h_a, h_c\}$
  are of degree at least $f$.  
This is because they are all heads in a gadget of type $\D$. 
Thus, the vertices of $A$, $C$, and $\{h_a, h_c\}$ have degree greater than
all the other vertices in $\G$. 
We refer to these vertices as {\em high degree} vertices and to the other
as {\em low degree} vertices.
\item There are $kn+m$ gadgets of type $\I$ and each gadget has $n$ vertices in the sets $A$ and $C$, respectively. 
Therefore, the number of high degree vertices in $\G$ is $2(n+1)(kn+m)$, and $k' = (n+1)(kn+m)$.
Also, the number of $\D$ type gadgets is $k'$. 
In the following points, we show that from each gadget exactly one head should be in $S$.  
\item If the set $S$ contains a low degree vertex, then it is possible to replace it with a high degree vertex which is not in $S$.
Since the edge probabilities are all identical, the resulting expected domination does not decrease.
\item Tails of the gadgets of type $\D$ are vertices with degree two, and thus are low-degree vertices. Therefore, the set $S$ does not contain any tails. 
\item Let $B$ be a block in $\G$. The vertices in the set $F(B)$ have degree $\max\{mn, n^2\}$, and are low-degree vertices.  Thus, we can assume that $S$ does not contain any vertex in $F(B)$.
\item The connector vertices are also low degree vertices
We conclude that $S$ does not contain a connector vertex.
\item There are $kn+m$ gadgets of type $\I$ and $k' = (n+1)(kn+m)$.
Therefore, $S$ contains $n+1$ vertices from each gadget of type $\I$.
Based on the observations above, it follows that $S$ contains vertices from
$A$, $C$, and $\{h_a, h_c\}$ from each gadget $I$ of type $\I$.
\item For every gadget of type $\D$, at least one of the head vertex must be in $S$. 
Suppose there exists a gadget $D$ such that both heads are not in $S$, then we cannot dominate the tail vertices of $D$. 
Since the number of gadgets of type $\D$ in $\G$ is the same as $k'$, there exists a gadget as $D'$ such that both heads are in $S$. 
Let us consider the set $S'$ obtained by replacing a head $\alpha$ in $V(D') \cap S$ by a head $\beta$ in $D$. Then, we get
$C(V(\G), S') \geq \C(V(\G), S) + fp-(f+n+k^2)(p-p^2)$ where
$fp-(f+n+k^2)(p-p^2) > 0$ since
$f > (n+k^2)/p$. 
This contradicts Observation~\ref{obs:maxtprime}, by which $\C(V(\G), S) = t'$
is the maximum value possible by a set of size $k'$.  


\item Since $S$ achieves an expected domination of $t'$, it follows that each gadget of type $\I$ selects exactly either $A \cup \{h_c\}$ or $C \cup \{h_a\}$ to achieve a part of the first term in the expression for $t'$.  Further, the additional term  of $2(1-(1-p)^n)$ for each gadget comes from covering the vertex named $d$ in a block $B$ which is adjacent to the set $A$ of each gadget in $B$, and a vertex named $b$ which is adjacent to the set $C$.
\end{itemize} 
We formalize the observations below.

\begin{claim}
For each tail vertex $x$ in $\G$, $N(x) \cap S$ is non empty.
\end{claim}
\begin{claim}
For every block $B$ in $\G$, and each gadget $I$ of type $\I$ in $B$, 
\\
$S \cap V(I) \subseteq A \cup C \cup \{h_a,h_c\}$. 
\end{claim}
\begin{claim}
\label{claim:AorC}
For every block $B$ in $\G$, and each gadget $I$ of type $\I$ in $B$, 
\\
either $A \cup \{h_c\} = S \cap V(I)$ or $C \cup \{h_a\} = S \cap V(I)$. 
\end{claim}
%
\begin{claim}
\label{claim:existUniqueGadget}
In every block $B$ in the graph $\G$, there exists a unique gadget $I$ such that $A \cup \{h_c\} = S \cap V(I)$. 
\end{claim}
\begin{claim}
\label{claim:coverageOfRest}
If the set $S$ satisfies Claim~\ref{claim:existUniqueGadget}, then $$\C(V(\G)\setminus R, S) = (kn+m)((n+1)fp + n+np+1+2(1-(1-p)^n)).$$
\end{claim}
\begin{proof}
Let $B$ be a block in $\G$. Let $I$ be a gadget of type $\I$ in $B$.  
Either the set $A$ or the set $C$ in $I$ is in $S$. 
Further, in every gadget of type $\D$ in $I$, exactly one head is in $A$ and another head is in $C$. The other two heads of a gadget of type $\D$ are the vertices in the set $\{h_a,h_c\}$. 
Therefore, in every gadget of type $\D$ in the graph $\G$, exactly one
of the heads is in $S$. 
For each pair $(u, v)$ such that there is a gadget $D_{u,v}$ in $I$, the expected domination of the set $V(D_{u,v})\setminus\{u,v\}$ by the set $S$ is $fp$. 
The gadget $I$ has $n+1$ gadgets of type $\D$. 
The expected domination of $V(I)$ by the set $S$ is given as follows:
\begin{eqnarray*}
\C(V(I), S) &=& \sum_{u,v \mid \D_{u,v} \in I}\C(V(D_{u,v})\setminus \{u,v\}, S) + \C(A \cup C \cup \{h_a, h_c\}, S)\\
&=& (n+1)fp + (n+1)+ np + (1-(1-p)^n)
\end{eqnarray*}
Then, the expected domination of $V(B)$ by the set $S$ is the sum of the expected domination contributed by the gadgets of type $\I$ and the domination due to $F(B)$ by the set $S$. 
There exists a unique gadget $I$ in the block $B$ such that the vertex set $A$ is added to $S$. 
In the remaining gadgets, the vertex set $C$ is added to $S$. 
We compute the value $\C(V(B), S)$ based on type of the block $B$
Let $B$ be a vertex-partition block $\G_{i}$ for some $i \in [k]$. 
All vertices in $F(\G_{i})$ except $b_{x_i}$ have $n$ neighbors in $S$. 
Therefore, the value $\C(V(\G_{i}), S)$ is given as follows:
\begin{eqnarray*}
\C(V(\G_{i}), S) &=& \sum_{\ell \in [n]}\C(V(\I_{\ell}), S) + \C(F(\G_{i}), S)\\
&=& n((n+1)fp + (n+1) + np + (1-(1-p)^n)) + n(1-(1-p)^n)
\end{eqnarray*}
Let $B$ be an edge-partition block $\G_{i,j}$ for some $1 \leq i,j \leq k$. 
All vertices in $F(\G_{i,j})$ except $b_{u_{i,x_i}u_{j,x_j}}$ have $n$ neighbors in $S$. 
Therefore, the value $\C(V(\G_{i,j}), S)$ is given as follows:
\begin{eqnarray*}
\C(V(\G_{i,j}), S) &=& \sum_{e \in E_{i,j}}\C(V(\I_e), S) + \C(F(\G_{i,j}), S)\\
&=& |E_{i,j}|((n+1)fp + (n+1) + np + (1-(1-p)^n)) + |E_{i,j}|(1-(1-p)^n)
\end{eqnarray*}
Finally, the expected domination of $V(\G)\setminus R$ by the set $S$ is computed as follows:
\begin{eqnarray*}
\C(V(\G)\setminus R, S) &=& \sum_{i \in [k]} \C(V(\G_i), S) + \sum_{1 \leq i < j \leq k} \C(V(\G_{i,j}), S)\\
&=& \sum_{i \in [k]}n((n+1)fp + (n+1) + np + 2(1-(1-p)^n))\\
&& + \sum_{1 \leq i < j \leq k} |E_{i,j}|((n+1)fp + (n+1) + np + 2(1-(1-p)^n))\\
&=& (m+kn)((n+1)fp + (n+1) + np + 2(1-(1-p)^n))
\end{eqnarray*}
Hence the claim is proved. 
\end{proof}
\begin{claim}
\label{claim:connectorCoverageUB}
If the set $S$  satisfies Claim~\ref{claim:existUniqueGadget}, then for each $1 \leq i < j \leq k$, $$\C(\{r_{i,j}^i, s_{i,j}^i\}, S) \leq 2(1-(1-p)^{n+1}).$$ 
\end{claim}
\begin{proof}
Since $S$ satisfies Claim~\ref{claim:existUniqueGadget}, it follows that for each block there is a unique gadget $I$ in the block such that the set $A$ inside the gadget $I$ is contained in $S$.
 For the vertex partition block $\G_i$, let $I_{x_i}$ be this unique gadget. Cleary, $x_i$ is a vertex in $V_i$.  Similarly, for the edge partition block $\G_{i,j}$, let $\I_{u_{i,z}u_{j,y}}$ be the corresponding unique gadget.  It is clear that $u_{i,z}u_{j,y} \in E_{i,j}$.
Let $\tilde{A}$ be the union of the sets $A$ in the    above mentioned gadgets. 
By construction of $\G$, the neighbors of the vertices $r_{i,j}^i$ and $s_{i,j}^i$ in $S$ are subsets of the set $\tilde{A}$. 
More precisely, $|N(s_{i,j}^i) \cap \tilde{A}| = x_i + n-z+1$ and $|N(r_{i,j}^i) \cap \tilde{A}| = n-x_i+1+z$.
Therefore, $$\C(\{r_{i,j}^i, s_{i,j}^i\}, S) = (1-(1-p)^{n + 1 + x_i - z}) + (1-(1-p)^{n + 1 + z - x_i}).$$
We consider two cases based on the values $x_i$ and $z$. 
First, consider the case $x_i \not= z$. Let $q = x_i-z > 0$.  
\begin{eqnarray*}
\C(\{r_{i,j}^i, s_{i,j}^i\}, S) &=& (1-(1-p)^{n + 1 + q}) + (1-(1-p)^{n + 1 - q})\\
&=& 2-\big((1-p)^{n+1}((1-p)^q + (1-p)^{-q})\big)\\
&<& 2-2(1-p)^{n+1}
~=~ 2(1-(1-p)^{n+1})
\end{eqnarray*}
Next, we consider $x_i = z$. Then, 
$$ \C(\{r_{i,j}^i, s_{i,j}^i\}, S) = (1-(1-p)^{n + 1 }) + (1-(1-p)^{n + 1}) = 2(1-(1-p)^{n + 1}).$$
The claim follows. 
\end{proof}


\subsection{Equivalence between multi-colored clique and Uni-PBDS}

\begin{lemma}
\label{lem:mcc:Uni-PBDS}
If $(G,k)$ is a YES-instance of the \mcc problem, then $(\G,k',t')$ is a YES-instance of the Uni-PBDS problem. 
\end{lemma}
\begin{proof}
Let $K = \{u_{i,x_i} \mid i \in [k]\}$ be a $k$-clique in $G$. 
That is, for each $i \in [k]$, the $x_i$-th vertex of the partition $V_i$
is in the clique. 
Now we construct a feasible solution $S$ for the instance $(\G,k',t')$ of the Uni-PBDS problem. 
The set $S$ consists of the following vertices.

For each $i \in [k]$,
\begin{itemize}
\item for each $\ell \in [n]$ with $\ell \not= x_i$, add $C \cup \{h_a\}$ 
in the gadget $\I_\ell$ in the vertex-partition block $\G_i$ to $S$, and
\item from the gadget $\I_{x_i}$ in the vertex-partition block $\G_i$, add $A \cup \{h_c\}$ to $S$.
\end{itemize}

For each $1 \leq i < j \leq k$, 
\begin{itemize}
\item for each edge $e \in E_{i,j}$ with $e \not= u_{i,x_i}u_{j,x_j}$,
add $C \cup \{h_a\}$ in the gadget $\I_e$ in the edge-partition block $\G_{i,j}$
to $S$, and 
\item for the edge $e=u_{i,x_i}u_{j,x_j}$, add $A \cup \{h_c\}$ in the gadget
$\I_e$ in the edge-partition block $\G_{i,j}$ to $S$. 
\end{itemize}

We show that $S$ is a feasible solution to the reduced instance of the Uni-PBDS.

\noindent
The expected domination by the set $S$ in $\G$ can be computed as follows: 
\begin{equation*}
\C(V(\G), S) = \C(R,S) + \sum_{i \in [k]}\C(V(\G_{i}), S) + \sum_{1 \leq i < j \leq k} \C(V(\G_{i,j}), S) 
\end{equation*}
By the definition of the set $S$, it satisfies the condition in Claim~\ref{claim:existUniqueGadget}.
Therefore, by Claim~\ref{claim:coverageOfRest},
\begin{eqnarray*}
\C(V(\G)\setminus R, S) &=& (m+kn)((n+1)fp + (n+1) + np + 2(1-(1-p)^n))\\
&=& t' - 4{k \choose 2}(1-(1-p)^{n+1}).
\end{eqnarray*}
\noindent
Next, the expected domination of the connector vertices by $S$ is computed as follows. 
For each $1 \leq i < j \leq k$, consider the pair of connectors $s_{i,j}^i$ and $r_{i,j}^i$. 
These vertices are connected to the set $A$ of every gadget of type $\I$ in the blocks $\G_i$ and $\G_{i,j}$. 
By the definition of $S$, there are unique gadgets $I_{x_i}$ in $\G_{i}$ and $I_{u_{i,x_i}u_{j,x_j}}$ in $\G_{i,j}$ for which the set $A$ is added to $S$.  Therefore, the neighbors of $s_{i,j}^i$ and $r_{i,j}^i$ in $S$ are the vertices from the set $A$ in these two gadgets. 
By construction of $\G$, both the vertices $s_{i,j}^i$ and $r_{i,j}^i$ have $n+1$ many neighbors in $S$. 
Therefore, the expected domination of the pair by the set $S$ is
\begin{eqnarray*}
\C(\{s_{i,j}^i, r_{i,j}^i\}, S) &=& 2(1-(1-p)^{n+1}).
\end{eqnarray*}
In the graph $\G$, there exist $2{k \choose 2}$ pairs of connector vertices.
Therefore,
\begin{eqnarray*}
\C(R, S) = 4{k \choose 2}(1-(1-p)^{n+1}).
\end{eqnarray*}
Based on the above calculations, the expected domination of $V(\G)$
by the set $S$ is
\begin{eqnarray*}
\C(V(\G), S) &=& \C(R, S) + \C(V(\G) \setminus R, S)\\
&=& 4{k \choose 2}(1-(1-p)^{n+1}) + t' - 4{k \choose 2}(1-(1-p)^{n+1})
~=~
t'
\end{eqnarray*}
Thus, the instance $(\G, k',t')$ is an $YES$-instance of the Uni-PBDS problem. 
\end{proof}


\begin{lemma}
\label{lem:Uni-PBDS:mcc}
If $(\G, k',t')$ is a YES-instance of the Uni-PBDS problem, then $(G,k)$ is a YES-instance of the \mcc problem. 
\end{lemma}
\begin{proof}
Let $S$ be a feasible solution to the instance $(\G, k', t')$ of the Uni-PBDS problem. 
For each $i \in [k]$, let $\I_{x_i}$ be the unique gadget for some $x_i \in [n]$, for which the set $A$ in $\I_{x_i}$ is in $S$. 
For each $1 \leq i < j \leq k$, let $\I_{u_{i,x_i}u_{j,x_j}}$ be the unique gadget for some $u_{i,x_i}u_{j,x_j} \in E_{i,j}$, for which the set $A$ in $\I_{u_{i,x_i}u_{j,x_j}}$ is in $S$. 
The existence of such gadgets are ensured by the Claim~\ref{claim:existUniqueGadget}. 
Let $K = \{u_{i,x_i} \mid i \in [k]\}$. 
We show that the set $K$ is a clique in $G$ as follows. 
Observe that we picked one vertex from each partition $V_i$ for $i \in [k]$. 
Next, we show that for each $1 \leq i < j \leq k$, there is an edge $u_{i,x_i}u_{j, x_j} \in E(G)$. 
Let $i,j$ such that $i < j$. 
By Claim~\ref{claim:coverageOfRest},
\begin{eqnarray*}
\C(V(\G)\setminus R, S) &=& (m+kn)((n+1)fp + (n+1) + np + 2(1-(1-p)^n))\\
&=& t' - 4{k \choose 2}(1-(1-p)^{n+1}).
\end{eqnarray*}
Since $\C(V(\G), S) \geq t'$ and $\C(V(\G)\setminus R, S) = t' - 4{k \choose 2}(1-(1-p)^{n+1})$, the expected domination of $R$ by $S$ is at least $4{k \choose 2}(1-(1-p)^{n+1})$. 
There are $2{k \choose 2}$  disjoint pairs of connector vertices in the graph $\G$. 
By Claim~\ref{claim:connectorCoverageUB}, each pair of connectors contributes at most an expected domination of value $2(1-(1-p)^{n+1})$.
It follows that for each pair of connector vertices the expected domination by $S$ is equal to $2(1-(1-p)^{n+1})$.  
Consequently, for each $i < j \in [k]$, each of the two pairs of connector
vertices connecting the blocks $\G_i$, $\G_{i,j}$ and $\G_j$, contributes
$2(1-(1-p)^{n+1})$ to the expected domination only if
$\I_{x_i}$, $\I_{u_{i,x_i}u_{j,x_j}}$, and $\I_{x_j}$ are the unique gadgets
in which the set $A$ of the gadgets are subsets of $S$.
It follows that $u_{i,x_i}u_{j,x_j} \in E(G)$.
Hence, the set $K$ forms a clique in $G$. 
\end{proof}



Given an instance $(G,k)$ of \mcc, the instance $(\G,k')$ is constructed in
polynomial time where $k'$ and $t'$ are polynomial in input size. 
By Lemma~\ref{lem:HPathwidth},
the pathwidth of $\G$ is a quadratic function of $k$.  
Finally, by Lemmas \ref{lem:mcc:Uni-PBDS} and \ref{lem:Uni-PBDS:mcc},
the Uni-PBDS instance  $(\G,k',t')$ output by the reduction is equivalent to
the \mcc instance $(G,k)$ that was input to the reduction. 
Since \mcc is known to be \wone-hard for the parameter $k$,
it follows that the Uni-PBDS problem is \wone-hard with respect to the
pathwidth parameter of the input graph. 
This complete the proof of Theorem~\ref{thm:whardpathwidth}.

\section{PBDS on Trees: PTAS and Exact Algorithm}
\label{sec:FPTAS}
In this section, we present our algorithmic results for the PBDS problem on trees.
Throughout this section, assume $\T$ is rooted at some vertex $r$. For each $x\in V$, 
denote by $\parent(x)$ the parent of $x$ in $V$, and by $\T(x)$ the subtree of $\T$ 
rooted at $x$.

\subsection{PTAS for PBDS on Trees}~\label{subsection:PTAS}
For each $v\in V$ and each $b\in[0,k]$, define ${\Y}_v(par,curr,b)$ to be the optimal value 
of $\C(V(\T(v)), S)$ where $par$ and $curr$ are boolean indicator variables that, respectively, 
denote whether or not $\parent(v)$ and $v$ are in $S$, and $b$ denotes the number of 
descendants of $v$ in~$S$. Formally, ${\Y}_v(par,curr,b)$ is represented as follows:
\begin{center}
$\arg\max\Big\{ \sum_{x\in \T(v)} \C(x,S)~~\Big|~~S\subseteq V,~|S\cap (\T(v)\setminus v)|=b, 
~curr=I_{v\in S}, ~par=I_{\parent(v)\in S}\Big\}$
\end{center}
The main idea behind our PTAS is to use the rounding method. Instead of computing ${\Y}_x$, 
we compute its approximation, represented as $\widehat{\Y}_x$. This is done in a bottom-up 
fashion, starting from leaf nodes of $\T$. For each $x\in V$, define $\delta(x)$ to be 
$|{\Y}_x - \widehat{\Y}_x|$. Throughout our algorithm, we maintain the invariant that 
$\widehat{\Y}_x\leq {\Y}_x$, for every $x\in V$.

We now present an algorithm to compute $\widehat{\Y}$.
Since ${\Y}_x$ is easy to compute for a leaf $x$, 
we set $\widehat{\Y}_x={\Y}_x$. 
For a leaf $x$,
\versiondense{
${\Y}_x(par,curr,b)$ is (i) undefined if $b\neq 0$, (ii) $\wt_x$ if 
$curr=1,b=0$, (iii) $\wt_x~p_{(\parent(x),x)}$
if $par=1,curr=0,b=0$, and (iv) zero otherwise.
}
\versionspacy{
$${\Y}_x(par,curr,b) =
\left\{
\begin{array}{ll}
\hbox{(i) undefined},  &  \hbox{if}~b\neq 0, 
\\
\hbox{(ii)}~ \wt_x,  &  \hbox{if}~curr=1,~b=0,
\\
\hbox{(iii)}~ \wt_x~p_{(\parent(x),x)},  &  \hbox{if}~par=1,~curr=0,~b=0,
\\
\hbox{(iv)}~ 0,  &  \hbox{otherwise}.
\end{array}
\right.
$$
}
Consider a non-leaf $v$. Let $z_1,\ldots,z_{t}$ be $v$'s children in $\T$, and $z_0$ be $v$'s 
parent in $\T$ (if exists). Let $\L(\beta)$, for $\beta\geq 0$, denote the collection of all integral 
vectors $\sigma=(b_1,curr_1,\ldots,b_t,curr_t)$ of length $2t$ satisfying
\versiondense{
(i) $curr_i\in\{0,1\}$ and $b_i\geq 0$, for $i\in[1,t]$, and
(ii) $\sum_{i\in[1,t]}(b_i+curr_i)=\beta$.
}
\versionspacy{
\\
(i) $curr_i\in\{0,1\}$ and $b_i\geq 0$, for $i\in[1,t]$, and
\\
(ii) $\sum_{i\in[1,t]}(b_i+curr_i)=\beta$.
\\
}
In our representation of $\sigma$ as $(b_1,curr_1,\ldots,b_t,curr_t)$,
the term $curr_i$ corresponds to the indicator 
variable representing whether or not $z_i$ lies in our tentative set $S$, and $b_i$ corresponds to 
the cardinality of $S\cap \big(V(\T(z_i))\setminus z_i\big)$. Further, for $i\in[1,t]$, let $\L_i(\beta)$ 
be the collection of those vectors $\sigma=(b_1,curr_1,\ldots,b_t,curr_t)\in \L(\beta)$ that satisfy 
$b_j,curr_j=0$ for $j>i$.

For a given $curr,par,b\geq 0$, we now explain the computation of $\widehat \Y_v(par,curr,b)$. 
Assume that we have already computed the approximate values $\widehat{\Y}_{z_i}$ $(i\in[1,t])$ 
corresponding to $v$'s children in $\T$. Setting $W = \max_{u \in V} \wt_u$, and using the scaling 
factor $M = \epsilon W/n$, let
\begin{align}
\quad
&A(\sigma) =\begin{cases*}
	\wt_v, & if curr=1,\\[1mm]
	M\bigg\lfloor\frac{\wt_v}{M}\Big(1- (1-par\cdot p_{(z_0,v)})\bigcdot
	\prod_{\substack{i\in[1,t] \\curr_i=1}}(1-p_{(z_i,v)}) \Big)\bigg\rfloor,  & otherwise,
	\end{cases*}\label{Eq:A_sigma}\\[2mm]
&B(\sigma)=\sum_{i\in[1,t]} \widehat{\Y}_{z_i}(curr,curr_i,b_i), \label{Eq:B_sigma}\\[2mm]
&\widehat{\Y}_v(par,curr,b) = \displaystyle \max_{\sigma\in \L(b)}\big(A(\sigma)+ B(\sigma)\big).
\end{align}

In order to efficiently compute $\widehat{\Y}_v$, we define the notion of {\em preferable} 
vectors. For any two vectors $\sigma_1,\sigma_2 \in \L(\beta)$, we say that $\sigma_1$ 
is {\em preferred} over $\sigma_2$ (and write $\sigma_1\geq \sigma_2$) if both
(i)~$A(\sigma_1)\geq  A(\sigma_2)$, and
(ii) $ B(\sigma_1)\geq  B(\sigma_2)$.
For $i\in[1,t]$, let $\L^*_i(\beta)$ be a {\em maximal} subset of $\L_i(b)$ such that 
$\sigma_1\ngeq \sigma_2$ for any two vectors $\sigma_1,\sigma_2\in \L^*_i(\beta)$. 

Define $\phi_v = |\{A(\sigma)~|~\sigma\in \L(\beta), \text{ for }\beta\in [0,k]\}|$.
The following observation is immediate by the definition of $\L^*_i$.

\begin{observation}\label{observation:size_Li}
For each $i\in[1,t]$ and $\beta\in[0,k]$, $|\L^*_i(\beta)|\leq \phi_v$.
\end{observation}

In order to compute $\widehat{\Y}_v(par,curr,b)$, we explicitly compute and store 
$\L^*_i(\beta)$, for $1\leq i\leq t$.  The set $\L^*_1(\beta)$ is quite easy to compute.
Let $\sigma_1=(\beta,0,0,\ldots,0)$ and $\sigma_2=(\beta-1,1,0,\ldots,0)$ be the only 
two vectors lying in $\L_1(\beta)$. Then $\L^*_1(\beta)$ is that vector among $\sigma_1$ 
and $\sigma_2$ that maximizes the sum $A(\sigma)+B(\sigma)$. 

The lemma below provides 
an iterative procedure for computing the sets $\L^*_i(\beta)$, for $i\geq 2$.

\begin{lemma}\label{lemma:computation_Li}
For every $i,\beta\geq 1$, the set $\L^*_i(\beta)$ can be computed from
$\L^*_{i-1}(\beta)$ in time
$\widetilde O\big(\beta+\sum_{\alpha\in[0,\beta]}~|\L^*_{i-1}(\alpha)|\big)$.
\end{lemma}

\begin{proof}
Initialize $\L^*_{i}(\beta)$ to $\emptyset$. At each stage, maintain the list $\L^*_{i}(\beta)$ sorted 
by the values $A(\cdot)$, and reverse-sorted by the values $B(\cdot)$. Our algorithm to compute 
$\L^*_{i}(\beta)$ involves the following steps.
\begin{enumerate}
\item For each $curr\in \{0,1\}$ and $b\in [0,\beta]$, first compute a set ${\cal P}_{b,curr}$ obtained 
by replacing the values $b_i$ and $curr_i$ in each $\sigma\in \L^*_{i-1}(\beta - (curr+b))$ by $b$ and 
$curr$ respectively. Let
\versiondense{
${\cal P}=\bigcup_{b\in[0,\beta],curr\in\{0,1\}} ~{\cal  P}_{b,curr}$.
}
\versionspacy{
$${\cal P} ~=~ \bigcup_{b\in[0,\beta],curr\in\{0,1\}} ~{\cal  P}_{b,curr}.$$
}
\item
For each $\sigma\in {\cal P}$, check in $O(\log |{\cal P}|)$ time if there is a $\sigma'\in \L^*_{i}(\beta)$ 
that is preferred over $\sigma$ (i.e. $\sigma' \geq \sigma$). If no such $\sigma'$ exists, then
\versiondense{
(a) add $\sigma$ to $\L^*_{i}(\beta)$, and
(b) remove all those $\sigma''$ from $\L^*_{i}(\beta)$ that are less preferred
than $\sigma$, that is, 
$\sigma''<\sigma$.
}
\versionspacy{
\\
(a) add $\sigma$ to $\L^*_{i}(\beta)$, and
\\
(b) remove all those $\sigma''$ from $\L^*_{i}(\beta)$ that are less preferred
than $\sigma$, that is, 
$\sigma''<\sigma$.
\\
}
\end{enumerate}
The runtime of the algorithm is $O(\beta+ |\cal P|\log |{\cal P}|)$ which is at most 
$\widetilde O(\beta+\sum_{\alpha\in[0,\beta]}~|\L^*_{i-1}(\alpha)|)$. 

Next we now prove its correctness. Consider a $\sigma=(b_1,curr_1,\ldots,b_t,curr_t)\in \L_{i}(\beta)$. It 
suffices to show that if $\sigma\notin {\cal P}_{b_i,curr_i}$, then there exists a $\sigma'\in {\cal P}_{b_i,curr_i}$ 
satisfying $\sigma'\geq \sigma$. Let $\sigma_0$ be obtained from $\sigma$ by replacing $b_i,curr_i$ with $0$. 
Since $\sigma\notin {\cal P}_{b_i,curr_i}$, it follows that $\sigma_0\notin \L_{i-1}(\beta-(b_i+curr_i))$. 
So there must exist a vector $\sigma'_0=(b'_1,curr'_1,\ldots,b'_t,curr'_t)\in \L_{i-1}(\beta-(b_i+curr_i)$ satisfying 
$A(\sigma'_0)\geq A(\sigma_0)$ and $B(\sigma'_0)\geq B(\sigma_0)$. Let $\sigma'$ be the vector obtained from 
$\sigma'_0$ by replacing $b'_i,curr'_i$ with $b_i,curr_i$. It can be easily verified from Eq.~\eqref{Eq:A_sigma} and~\eqref{Eq:B_sigma}, that $A(\sigma')\geq A(\sigma)$ and $B(\sigma')\geq B(\sigma)$. 
Since the constructed $\sigma'$ indeed lies in ${\cal P}_{b_i,curr_i}$,
the proof follows.
\end{proof}

The following claim is an immediate corollary of Lemma~\ref{lemma:computation_Li}.

\begin{lemma}\label{lemma:computationYv}
The value of $\widehat{\Y}_v(par,curr,b)$, for any $par,curr\in \{0,1\}$ and
$b\in[0,k]$, is computable in 
$\widetilde O(b\cdot\deg(v)\cdot\phi_v)$ time, given the values
of $\widehat{\Y}_{z_i}$ for $i\leq t$.
\end{lemma}

\begin{proof}
Observe that
$\widehat{\Y}_v(par,curr,b)=\max_{\sigma\in\L^*_t(b)}\big(A(\sigma)+B(\sigma)\big)$,
where $A(\sigma)$ and $B(\sigma)$ are as defined in Eq.~\eqref{Eq:A_sigma}
and~\eqref{Eq:B_sigma}.
By Observation~\ref{observation:size_Li} and Lemma~\ref{lemma:computation_Li},
the total computation time of the set $\L^*_t(b)$ is at most
$\widetilde O(b\cdot t\cdot\phi_v)$, which is equal to
$\widetilde O(b\cdot\deg(v)\cdot\phi_v)$.
%
\end{proof}

Lemma~\ref{lemma:computationYv} implies that starting from leaf nodes,
the values $\widehat{\Y}_x(par,curr,b)$ can be computed in bottom-up manner,
for each valid choice of triplet $(par,curr,b)$ and each $x\in V$,
in total time $\widetilde O(k^2n\cdot \max_{v\in V}\phi_v)$. We now prove
$\phi_v=O(\epsilon^{-1}n)$. If $curr=1$, then $A(\sigma)$ takes only one value.
If $curr=0$, then the value of $A(\sigma)$ is a multiple of $M$ and is also bounded above by $W$.
This implies that the number of distinct values $A(\sigma)$ can take
is indeed bounded by $W/M=O(\epsilon^{-1}n)$.

\begin{proposition}\label{lemma:total_runtime}
Computing $\widehat{\Y}_x$ for all $x\in V$ takes in total
$\widetilde O(k^2n\cdot \max_{x\in V}\phi_x)=\widetilde O(k^2\epsilon^{-1}n^2)$ time.
\end{proposition}

\subsection{Approximation Analysis of PTAS on Trees}
\label{appendix:approxAnalysis}

We provide here the approximation analysis of the $(1-\epsilon)$-bound.
Let
\begin{align*}
	S_\opt &~= \argmax_{{S\subseteq V, |S|=k}}~\C(V, S) ~= \argmax_{{S\subseteq V, |S|=k}}~
		\sum_{x\in V}\wt_x\Pr(x \sim S),\\
	\widehat S_\opt &~= \argmax_{{S\subseteq V, |S|=k}}~
		\bigg(\sum_{x\in S} \wt_x\Pr(x\sim S) + \sum_{x\in V\setminus S}M\Big\lfloor 
		\frac{\wt_x\Pr(x\sim S)}{M} \Big\rfloor \bigg).
\end{align*}

Observe that
\versiondense{
$\max\{{\Y}_r(0,0,k),{\Y}_r(0,1,k-1)\}=\C(V, S_\opt)$ and 
$\max\{{\widehat \Y}_r(0,0,k),{\widehat\Y}_r(0,1,k-1)\}=\C(V,\widehat S_{opt})$.
}
\versionspacy{
\\
$\max\{{\Y}_r(0,0,k),{\Y}_r(0,1,k-1)\}=\C(V, S_\opt)$ ~~and
\\
$\max\{{\widehat\Y}_r(0,0,k),{\widehat\Y}_r(0,1,k-1)\}=\C(V,\widehat S_{opt})$.
\\
}
The following lemma proves that $\widehat S_\opt$ indeed achieves a $(1-\epsilon)$-approximation bound.

\begin{lemma}
\label{lem:apxlowerbound}
$(1-\epsilon)~\C(V, S_\opt)~\leq~ \C(V, \widehat S_\opt)~\leq ~\C(V, S_\opt)$.
\end{lemma}

\begin{proof}
In order to prove the first inequality, we first show that
\\
$\C(V, S_{opt}) - \C(V, \widehat S_{opt})\leq~ \epsilon~\C(V, S_{opt})$.
\begin{align*}
	\C(V, S_{opt}) - \C(V, \widehat S_{opt})
	&\leq
	\max_{\substack{S\subseteq V\\ |S|=k}}~
		\Big(\sum_{x\in V}\wt_x\Pr(x\sim S) - \sum_{x\in S} \wt_x\Pr(x\sim S)
		- \sum_{x\in V\setminus S}M\Big\lfloor \frac{\wt_x\Pr(x\sim S)}{M} \Big\rfloor \Big)\\
	&\leq (n-k) M ~\leq~ \epsilon~ W~\leq~ \epsilon~\C(V, S_{opt}).
\end{align*}
Next, for each $x\in V$ and $S\subseteq V$, we have $M\lfloor M^{-1}\wt_x\Pr(x\sim S)\rfloor\leq \wt_x\Pr(x\sim S)$, 
thereby implying that $\C(V, \widehat S_\opt)\leq \C(V, S_\opt)$. This completes our proof. 
\end{proof}

For any integer $k$, any $n$-vertex tree $\T$ with arbitrary edge probabilities, and for every $\epsilon > 0$, a $(1-\epsilon)$ approximate solution can be computed in time $\widetilde O(k^2\epsilon^{-1}n^2)$. 
This follows from Proposition~\ref{lemma:total_runtime} and Lemma~\ref{lem:apxlowerbound}.
We thus prove Theorem~\ref{thm:FPTAS}.


\subsection{Linear-time algorithm for Uni-PBDS on Trees}
We next establish our result for the scenario of Uni-PBDS on trees
(Theorem~\ref{thm:boundedProb}).
In fact, this result holds for a somewhat broader scenario,
wherein, for each vertex $x$, the cardinality of
$\textsc{prob}_x=\{p_e~|~e$~is incident to~$x\}$ is bounded above
by some constant $\gamma$.

\begin{proof}[Proof of Theorem~\ref{thm:boundedProb}]
Observe that the only place where approximation was used in our PTAS was in bounding the number of distinct values that can be taken by $A(\sigma)$ in Eq.~\eqref{Eq:A_sigma}. In order to obtain an exact solution for the bounded probabilities setting, the only modification performed in our algorithm is to redefine $A(\sigma)$ as follows.

\begin{align*}
A(\sigma) &=  \wt_v\cdot I_{(curr=1)}~+~\wt_v\Big(1- (1-par\cdot p_{(z_0,v)})\bigcdot\prod_{\substack{i\in[1,t] \\curr_i=1}}(1-p_{(z_i,v)}) \Big)\cdot I_{(curr\neq 1)}.
\end{align*}

It can be verified that the algorithm correctly computes $\Y_x$ at each step, that is, $\delta(x)$ is essentially zero.
The time it takes to compute $\Y_v(par,curr,b)$, for a non-leaf $v$, crucially depends on the cardinality of
$\{A(\sigma)~|~\sigma\in \L_t^*(b)\}$, where $t$ is the number of children of $v$ in $\T$.
Observe that the number of distinct values $A(\sigma)$ can take is at most $b^{|\textsc{prob}_x|}=O(k^\gamma)$. This along with Lemma~\ref{lemma:total_runtime} implies that the total runtime of our exact algorithm is $\widetilde O(k^{\gamma+2}n)$.
\end{proof}

\subsection{Solving PBDS optimally on general trees}

Let $c\geq 1$ be the smallest real such that
$2$-SPM problem has an $\widetilde O(N^{c})$ time algorithm. 
We will show that in such a case, $k$-PBDS can be solved optimally on trees 
with arbitrary probabilities in $\widetilde O((\delta N)^{c\lceil k/2\rceil +1})$ time, for a constant $\delta >0$.

For any node $v\in \T$, let $\T_v^i$, for $1\leq i\leq \deg(v)$, represent the 
components of the subgraph $\T\setminus\{v\}$. 
We start with the following lemma
which is easy to prove using a standard counting argument.

\begin{lemma}~\label{lemma:counting}
For any set $S$ of size $k$ in $\T$, there exist a node $v\in \T$
and an index $q\in[1,\deg(v)]$ such that
the cardinalities of the sets $S\cap \big(\bigcup_{i\lneq q} T_v^i \big)$,
$S\cap \big(T_v^q \big)$ and $S\cap \big(\bigcup_{i\gneq q} T_v^i \big)$
are all bounded by $k/2$.
\end{lemma}

\begin{proof}
We first show that there exists a node $v$ in $\T$ satisfying the property 
$|S\cap T_v^i|\leq k/2$, for each $i\in[1,\deg(v)]$.
Consider a node $u\in \T$. If $u$ satisfies the above mentioned property then we are done.
Otherwise, there exists an index $j\in[1,\deg(u)]$ for which
$|S\cap T_u^{j}|\gneq k/2$.
This implies the number of elements of $S$ lying in 
$\{u\}\cup\big(T_u^1 \cup \cdots \cup T_u^{j-1} \big)\cup\big(T_u^{j+1} \cup \cdots \cup T_u^{\deg(u)} \big)$
is at most $k/2$.
In such a case we replace $u$ by its $j^{th}$ neighbor.
Repeating the process eventually leads to the required node $v$.

Now, let $q\in[1,\deg(v)]$ be the smallest integer for which
$S\cap\big(T_v^1 \cup \cdots \cup T_v^{q} \big)$ is larger than $k/2$.
Then, $S\cap \big(\bigcup_{i\lneq q} T_v^i \big)$ and
$S\cap \big(\bigcup_{i\gneq q} T_v^i \big)$ are both bounded by $k/2$,
by definition of $q$.
Also, $S\cap T_v^q$ is bounded by $k/2$ due to the choice of $v$.
\end{proof}

For the rest of this section, we refer to a tuple $(v,q)$ satisfying
the conditions stated in Lemma~\ref{lemma:counting} as a {\em valid pair}.
Let us suppose we are provided with a valid pair $(v,q)$. 
For sake of convenience, we assume that $\T$ is rooted at node $v$.
Let $U_0=T_v^q$, $U_1=\bigcup_{i\lneq q} T_v^i$, and $U_2=\bigcup_{i\gneq q} T_v^i$.
Also let $\{x_1,\ldots,x_\alpha\}$ be the children of $v$ in $U_1$, where
$\alpha=q-1$;
let $\{y_1,\ldots,y_\beta\}$ be the children of $v$ in $U_2$, where $\beta=\deg(v)-q$;
and let $z$ be $q^{th}$ child of $v$.

An important observation is that if the optimal $S$ contains $v$, then the problem is easily solvable in $O(k^2\cdot n^{k/2})$ time, as then the structures
$U_0$, $U_1$ and $U_2$ become independent.
Indeed, it suffices to consider all $O(k^2)$ partitions of $k-1$ into triplet $(c_0,c_1,c_2)$ consisting of integers 
in the range $[0,k/2]$, and next we iterate over all the $c_i\leq k/2$ subsets in $U_i$.
This takes in total $O(k^2 n^{k/2})$ time.
So the challenge is to solve the problem when $v$ is not contained in $S$.
Assuming $(v,q)$ is a valid pair, and $v$ is not contained in the optimal $S$,
the solution $S$ to $k$-PBDS is the union $S_0\cup S_1\cup S_2$
of the tuple $(S_0,S_1,S_2)\in U_0\times U_1\times U_2$ that maximizes
\begin{equation}~\label{eq:exact-pbds}
\C(U_0,S_0)+\C(U_1,S_1)+\C(U_2,S_2)+
\wt_v\Big(1- (1-d\cdot p(v,z))\bigcdot
\prod_{\substack{i\in[1,\alpha]\\x_i\in S_1}} (1-p(v,x_i))
\prod_{\substack{j\in[1,\beta]\\y_i\in S_2}} (1-p(v,y_j))
\Big)
\end{equation}
where
$d$ is an indicator variable denoting whether or not $z\in S_0$, and
$|S_1|,|S_2|,|S_3|$ are integers in the range $[0,k/2]$ that sum up to $k$.

Define $\Gamma$ to be set of all quadruples $(d,c_0,c_1,c_2)$ comprising of
integers in the range $[0,k/2]$ 
such that $d\in \{0,1\}$ and $c_0+c_1+c_2=k$. 
For each $\gamma=(d,c_0,c_1,c_2)\in \Gamma$, let

\begin{align*}
\L_\gamma^1 ~&=~ \Big\{ \Big(~\frac{\C(U_1,S_1)}{\wt_v(1-d\cdot p(v,z))}, \prod_{\substack{i\in[1,\alpha]\\x_i\in S_1}} (1-p(v,x_i))~\Big)
~\Big |~S_1\subseteq U_1 \text{ is of size }c_1 \Big\}~,\\
\L_\gamma^2 ~&=~ \Big\{ \Big(~\frac{\C(U_2,S_2)}{\wt_v(1-d\cdot p(v,z))}, \prod_{\substack{j\in[1,\beta]\\y_j\in S_2}} (1-p(v,y_j))~\Big)
~\Big |~S_2\subseteq U_2 \text{ is of size }c_2 \Big\}~,\\
Z_\gamma~&=~ \max \Big\{\C(U_0,S_0)~\Big|~S_0\subseteq U_0 \text{ is of size }c_0,\text{ and }d=I_{z\in S_0} \Big\}~.
\end{align*}
So, the maximization considered in Eq.~(\ref{eq:exact-pbds}) is equivalent to
the following optimization:
\begin{equation}\label{eq:exact-pbds-spm}
\max_{\substack{\gamma=(d,c_0,c_1,c_2)\in\Gamma\\(a,b)~\in~\L^1_\gamma,~(\bar a,\bar b)~\in~\L^2_\gamma}}~
 \wt_v+Z_\gamma + \Big(\wt_v(1-d\cdot p(v,z)) \Big)
\Big(a+\bar a- b\bar b\Big)~.
\end{equation}
In the next lemma we show that optimizing the above expression is equivalent to 
solving $|\Gamma|=O(k^2)$ different $2$-SPM problems
(each of size $O(n^{k/2})$).

\begin{lemma}
\label{lemma:2-spm-conversion}
Let $A=\big((a_1,b_1),\ldots,(a_N,b_N)\big)$ and 
$\bar A=\big((\bar a_1,\bar b_1),\ldots,(\bar a_N,\bar b_N)\big)$
be two arrays. Then, solving the maximization problem
$\max_{i_0,j_0}(a_{i_0}+\bar a_{j_0}-b_{i_0}\bar b_{j_0})$,
can be transformed in linear time to the following equivalent $2$-{\sc SPM}:
$$\L=\big((Q+a_1,Rb_1),\ldots,(Q+a_N,Rb_N),
(-Q+\bar a_1, R^{-1}\bar b_1),\ldots,(-Q+\bar a_N,R^{-1}\bar b_N)
\big)~,$$
where $Q=\max_{i,j} (b_i\bar b_j) + 2\max_{i,j}(a_i+\bar a_j)$ and 
$R =  \sqrt{4Q/\min_{i}(b_i)^2}$.
\end{lemma}

\begin{proof}
Consider the following $2$-SPMs obtained by two equal partitions of $\L$:
\begin{align*}
\L^1~&=~\big((Q+a_1,Rb_1),\ldots,(Q+a_N,Rb_N)\big)~\text{, and}\\
\L^2~&=~(-Q+\bar a_1, R^{-1}\bar b_1),\ldots,(-Q+\bar a_N,R^{-1}\bar b_N)\big)~.
\end{align*}

Observe that the optimal value of $\L^1$ is bounded above by
$2Q+2\max_{i}(a_i)-R^2\min_i(\bar b_i)^2$ which~is~strictly less than $-Q$.
Similarly, the optimal value of $\L^2$ is bounded above by
$2Q+2\max_{i}(\bar a_i)-\min_i(\bar b_i)^2/R^2$,
which is again strictly less than $-Q$.

Now the answer to the optimization problem
$\max_{i_0,j_0}(a_{i_0}+\bar a_{j_0}-b_{i_0}\bar b_{j_0})$ 
is at least $-Q$.
This clearly shows that the solution to $\L$ cannot be obtained by its restrictions $\L^1$ and $\L^2$.
Hence, the maximization problem
$\max_{i_0,j_0}(a_{i_0}+\bar a_{j_0}-b_{i_0}\bar b_{j_0})$ 
is equivalent to solving the $2$-SPM $\L$.
\end{proof}

\begin{proof}[Proof of Theorem~\ref{thm:twospmforpbds}]
The time to compute $\L^1_\gamma$, $\L^2_\gamma$, $Z_\gamma$,
for a given $\gamma$, is $\tilde O(n^{k/2})$.
By transformation presented in Lemma~\ref{lemma:2-spm-conversion},
it follows that the total time
required to optimize the expression in
Eq.~(\ref{eq:exact-pbds-spm}) is
$k^{O(1)}\cdot n^{c\lceil k/2\rceil }$,
which is at most $O\big((\delta n)^{c\lceil k/2\rceil +1}\big)$, for some constant $\delta\geq 1$.
Now recall that Eq.~(\ref{eq:exact-pbds-spm})
provides an optimal solution assuming
$(v,q)$ is a valid pair, and $v$ is not contained in optimal $S$.
Even when $(v,q)$ is a valid pair, and $v$ is contained in the optimal $S$,
the time complexity turns out to be $O(k^2 \cdot n^{k/2})$.
Iterating over all choices of pair $(v,q)$
incurs an additional multiplicative factor of $n$ in the runtime.
\end{proof}

\section{Parameterization based on graph structure}
\label{section:fpt-algorithms}
In this section, we state our results on structural parameterizations
of the Uni-PBDS problem.
First,
following the approach of Arnborg et al.~\cite{ArnborgLS91},
we formulate the MSOL formula for the Uni-PBDS problem where the
quantifier rank of the formula is $\Oh(k)$. 
This indeed yield an FPT algorithm for the Uni-PBDS problem parameterized by
budget $k$ and the treewidth of the input graph.

In addition, we show that the Uni-PBDS problem is FPT for the budget parameter
on apex-minor-free graphs.
In particular we show that, for any integer $k$, and any $n$-vertex weighted apex-minor free graph 
with uniform edge probability, the Uni-PBDS problem can be solved in time $(2^{\Oh(\sqrt{k} \log k)}n^{\Oh(1)})$.

\subsection{Results on Planar and Apex Minor-Free graphs}~\label{appendix: apex-minor-free}
We present here a subexponential time algorithm to solve the Uni-PBDS problem on planar and apex minor-free graphs.
The algorithm is based on the technique due to Fomin~et~al.~\cite{FominLRS11} used in the subexponential algorithm for the partial cover problem, and the claim proved in Theorem~\ref{thm:twkfpt} and Theorem~\ref{theorem:fpt-by-wk}. 
Let $H$ be a given apex graph. Then our input is an instance $\langle \G=(V, E, p, \wt), k \rangle$ of the {\sc PBDS} problem where $G=(V,E)$ is an $H$-minor-free graph.
Let $\sigma = (v_1,v_2,\ldots,v_n)$ be an ordering of the vertices in
non-increasing order of their expected coverage, that is,
$\C(V, v_1) \geq \C(V, v_2) \geq \cdots \geq \C(V, v_n)$. 
For $1 \leq i \leq n$, let $V_\sigma^i = \{v_1,\ldots,v_i\}$ and $\G_\sigma^i = \G[V_\sigma^i]$. 
Let $S_{opt}$ be an optimal solution and $i$ be the largest index of a vertex in $S_{opt}$.
The following lemma states a crucial property of $S_{opt}$.
\begin{lemma}\label{lemma:lexileast}
$S_{opt}$ is a $3$-dominating set for $\G[N[V_\sigma^i]]$.
\end{lemma}
 \begin{proof}
 It suffices to show that $S_{opt}$ is a $2$-dominating set for $\G_\sigma^{i}$. The proof is by contradiction. Suppose $S_{opt}$ is not a $2$-dominating set for $\G_\sigma^{i}$, then there exists a vertex $v_j \in V_\sigma^{i}$, with $j <{i}$, such that $N_{\G_\sigma^{i}}[S_{opt}] \cap N_{\G_\sigma^{i}}[v_j] = \emptyset$. Let $S' = (S_{opt}\setminus\{v_{i}\})\cup \{v_j\})$. We know that $\C(V, v_j) \geq \C(V, v_{i})$, and $v_j$ is not 2-dominated by $S$. Thus, 
 $$\C(V, S') = \C(V, S_{opt}\setminus\{v_{i}\}) + \C(V, v_j) \geq \C(V, S_{opt}) - \C(V, v_{i}) + \C(V, v_j) \geq \C(V, S_{opt}).$$
Clearly, $S'$ is also an optimal solution and $S'$ is lexicographically smaller than $S_{opt}$. This contradicts the fact that $S_{opt}$ is the lexicographically least solution. Therefore, the set $S_{opt}$ must be a $2$-dominating set for $\G_\sigma^{i}$, and  thus also a $3$-dominating set for $\tildeG_\sigma^{i}$.
 \end{proof}

We use the following structural property on apex-minor-free graphs from
Fomin et al.~\cite{FominLRS11}.

\begin{lemma}[Fomin et al.~\cite{FominLRS11}]\label{lemma:bdrdom}
If an apex-minor-free graph $G$ has an $r$-dominating set of size $k$ then the treewidth of the graph $G$ is at most $(c\bigcdot r\sqrt{k})$, where $c$ is a constant dependent only the size of the apex graph.
\end{lemma}

We next use the following lemma, to compute an approximation to treewidth of  prescribed-minor-free graphs.
\begin{lemma}[Demaine et al.~\cite{DemaineHK05}]
\label{lemma:tdaprox}
For each fixed $H$, there is a polynomial time algorithm, which 
for any $H$-minor free graph $G$ computes a tree decomposition of width  $\delta$ times the treewidth of $G$, where $\delta$ is a constant.
\end{lemma}

Algorithm \ref{algorithm:fpt-amfg}, presented next, solves the {\sc PBDS}
problem on apex minor free graphs.
\begin{algorithm}[h]
 \KwData{An uncertain $H$-minor-free graph $\G=(V,E, p, \wt)$, 
 and an integer~$k$, where $H$ is an apex graph.}
For each $j\in[n]$, compute the $\delta$-approximate treewidth $tw_j$ of $ \G[N[V_\sigma^j]]$ using  Lemma~\ref{lemma:tdaprox}.\\
Let $i=\max\{ j~|~ tw_j \leq 3c\bigcdot\delta\bigcdot\sqrt{k}\}$.\\
Compute a tree decomposition $(\TD,X)$ of $\G[N[V_\sigma^i]]$ using Lemma~\ref{lemma:tdaprox}.\\
Run the FPT algorithm in Theorem \ref{thm:twkfpt} on the instance $\langle \G[N[V_\sigma^i]],p,w, k \rangle$ with tree decomposition $(\TD,X)$ and output the solution.
 \caption{FPT for the {\sc PBDS} problem in apex-minor-free graphs}
\label{algorithm:fpt-amfg}
\end{algorithm}

\begin{theorem}\label{theorem:fpt-amfg}
For any integer $k$, and any $n$-vertex weighted apex-minor free graph $\G$ 
with uniform edge probability, the Uni-PBDS problem can be solved in time $(2^{\Oh(\sqrt{k} \log k)}n^{\Oh(1)})$.
\end{theorem}

\section{Uni-PBDS problem parameterized by the treewidth and budget $k$}

\newcommand{\qr}{\text{qr}}
\subsection{{MSOL formulation of the Uni-PBDS problem}}\label{appendix:mso}
We show that an extension of Courcelle's theorem due to Arnborg et. al.~\cite{ArnborgLS91} results in an FPT algorithm for the combined parameters treewidth and $k$. 
This is obtained by expressing the Uni-PBDS problem as a {\em monadic second order logic} (MSOL) formula (see ~\cite{Courcelle90,CourcelleEngelfriet12}) of length $\Oh(k)$.
 The following MSOL formulas  are used in the MSOL formula for the Uni-PBDS problem. 
The upper case variables (with subscripts) take values from the set of subsets of $V$, and the lower case variables take values from $V$.
\begin{itemize}
\item The vertex set $S$ contains $d$ elements. 
\begin{eqnarray*}
{\tt SIZE}_d (S) = \exists x_1 \exists x_2 \cdots \exists x_d \forall y ( y\in S \to \bigvee_{i=0}^d ( y = x_i)) 
\end{eqnarray*}
\item Given a vertex set $X$ and a vertex $x$, there exists a set $S \subseteq X$ of size $d$, and for each vertex $y$ in $S$, $y$ is a neighbor of $x$. 
\begin{eqnarray*}
{\tt INC}_d (x, X) = \exists S({\tt SIZE}_d(S) \wedge \forall y ((y \in S \to y \in X) \wedge  (y\in S \to adj(x,y)))) 
\end{eqnarray*}
\item The sets $X$, $Y_0, Y_1,\ldots,Y_k$ partition the vertex set $V$. 
\begin{eqnarray*}
{\tt PART}(X, Y_0, Y_1,\ldots,Y_k) = \forall x \Big( && (x \in X \vee \displaystyle\bigvee_{i=0}^k x \in Y_i) \wedge \\
&& \displaystyle\bigwedge_{i=0}^k \neg(x \in X \wedge x \in Y_i) \wedge \bigwedge_{i\not= j} \neg(x \in Y_i \wedge x \in Y_j)\Big)\\
\end{eqnarray*}
\end{itemize}
Now we define the MSOL formula for the Uni-PBDS problem.  The formula expresses the statement that $V$ can be partitioned into $X$ and $V\setminus X$, and $V\setminus X$ can be partitioned into $k+1$ sets $Y_0, Y_1, \ldots, Y_k$ such that for each set $Y_i$ and each vertex $y$ in $Y_i$, $y$ has $i$ neighbors in $X$. 
\begin{eqnarray*}
\text{{\tt Uni-PBDS}} = \exists X \exists Y_0 \exists Y_1 \cdots \exists Y_k \bigg(&&{\tt PART}(X, Y_0, Y_1, \ldots, Y_k) \wedge \\
&& \forall x \forall y \Big(\big(y \in Y_0 \wedge x \in X) \to \neg\big(adj(x,y)\big)\Big)\\
&& \forall y \Big( \bigwedge_{i = 1}^k \big(y \in Y_i \to {\tt INC}_i(y, X)\big)\Big)\bigg)
\end{eqnarray*}

\begin{lemma}
\label{lem:qr}
The {\em quantifier rank} of {\tt Uni-PBDS} is $\Oh(k)$.
\end{lemma}
\begin{proof}
There are $k+2$ initial quantifiers for the sets $X, Y_0, Y_1,\ldots, Y_k$. For two MSOL formulas $\phi$ and $\psi$ with quantifier rank $\qr(\phi)$ and $\qr(\psi)$, respectively, $\qr(\phi \wedge \psi) = \qr(\phi \vee \psi) = \max\{\qr(\phi), \qr(\psi)\}.$ Therefore, \qr({\tt Uni-PBDS}) is bounded as follows:
\begin{eqnarray*}
\qr(\text{\tt Uni-PBDS})&=& k+2 + \max\{\qr(\text{\tt PART}), 1+\qr(\text{\tt INC})\}\\
&=& k+2 + \max\{1, 1+\max\{\qr(\text{\tt SIZE}), 2\}\}\\
&\leq& k+2 + k = 2k + 3 = \Oh(k)
\end{eqnarray*}
\end{proof}

\noindent
We now show that the {\tt Uni-PBDS} problem
is fixed-parameter tractable
in parameters $k$ and treewidth by expressing the maximization problem on the MSOL formula as a minor variation of {\em extended monadic second-order extremum problem} as described by Arnborg et. al. \cite{ArnborgLS91}.  

\begin{proof}[Proof of Theorem~\ref{thm:twkfpt}]

For each $0 \leq i \leq k$, define the weight function $w^i$ associated with the set variable $Y_i$ as follows: for each $v \in V(G)$, 
$w_v^i = (1-(1-p)^i)w(v)$.  The difference between the weight function in \cite{ArnborgLS91} and our problem is that in their paper $w(v)$ is considered to be constant value, for all vertices, for the set variable $Y_i$.
Observe, however, that the running time of their algorithm does not change as long as $w_v^i$ can be computed in polynomial time, which is the case in our definition.  
Therefore, our maximization problem is now formulated as a variant of the EMS maximization problem in \cite{ArnborgLS91}:
\[Maximize \sum_{u \in X} w(u) + \sum_{i = 0} ^ k \sum_{u \in Y_i} w_v^i\cdot y_v^i \text{ over partitions } (X,Y_0,Y_1,\ldots,Y_k)\text{  satisfying } \text{\tt Uni-PBDS}\]
Using Theorem 5.6 in \cite{ArnborgLS91} along with the
additional observation, we make, that $w_v^i$ can be efficiently computed, an optimal solution for the Uni-PBDS problem can be computed in time
$f(\qr(\text{\tt Uni-PBDS}),w)\cdot poly(n)$, where $f(\qr(\text{\tt Uni-PBDS}),w)$ is a function which does not depend on $n$- it depends only on the quantifier rank of Uni-PBDS and the treewidth.
By Lemma \ref{lem:qr},
\qr(\text{\tt Uni-PBDS}) $=O(k)$, and thus by \cite{ArnborgLS91},
$f(\qr(\text{\tt Uni-PBDS}),w) = f(O(k), w)$.  This shows that Uni-PBDS is FPT with resepect the parameters $k$ and treewidth. Hence the theorem is proved.
\end{proof}

\subsection{An Alternative Parameterization for Uni-PBDS by treewidth and budget}
\label{sec:fpt-dp}


We present here an alternate algorithm for Uni-PBDS parametrized by treewidth and budget.
Recall that in the previous section we showed using Courcelle's theorem 
that for any integer $k$, and any $n$-vertex uncertain graph of treewidth $w$
with uniform edge probabilities, 
$k$-Uni-{\sc PBDS}
can be solved in time $O(f(k,w) n^2)$.
Here, we show that the Uni-PBDS problem for
any $n$-vertex uncertain graph of treewidth $w$
can be solved in $2^{\Oh(w \log k)}n^2$ time, thereby proving that  $f(k,w)=2^{\Oh(w \log k)}$.\\

 \noindent{\bf Definitions and Notation}. We first set up some definitions and notation required for the DP formulation.
Let $(\TD, X)$ be a nice tree decomposition of the graph $\G$ rooted at node $\blr$. 
For a node $\bli\in V(\TD)$, let $\TD_\bli$ be the subtree rooted at $\bli$ and
$\Xi^+ = \cup_{\blj \in  V(\TD_\bli)}\{\Xj\}$. 
The uncertain graph induced by the vertices $\Xi^+$ is $\G[\Xi^+]$ and it is denoted by $\Gi$.  
The expected domination function $\C$ over the graph $\Gi$ is denoted by $\Ci$.  We refer to the expected domination as coverage in the presentation below.

For each node $\bli \in \TD$, we compute two tables $\soli$ and $\vali$. 
The rows of both tables are indexed by 4-tuples which we refer to as states. $\S_\bli$ denotes the set of all states associated with node $\bli$.  For a state $s$, the DP formulation gives a recursive definition of the values $\soli[s]$ and $\vali[s]$.  $\soli[s]$ is a subset $S$ of $\Xi^+$ that achieves the optimum coverage for $\Gi$ and satisfies additional constraints specified by the state $s$. $\vali[s]$ is the value $\Ci(\Xi^+,\soli[s])$.  
A state $s$ at the node $\bli$ is a tuple $(b, \gamma, \alpha, \beta)$, where
\begin{itemize}
\item $0 \leq b \leq k$ is an integer and specifies the size of $\soli[s]$,
\item $\gamma:\Xi\to \{0,1\}$ is an indicator function for the vertices of $\Xi$. This specifies the constraint that
$\gamma^{-1}(1) \subseteq \soli[s]$ and $\gamma^{-1}(1) \cap \soli[s] = \emptyset$.  We use  $A$ to denote $\cg\inv(1)$ and the state will be clear from the context.
\item $\alpha:\Xi\to [0,k]$ is a function. The constraint is that for each vertex $u \in \gamma\inv(0)$, $u$ should have $\alpha(u)$ neighbors in $\soli[s]$.
\item $\beta: \Xi \to [0,k]$ is function. The constraint is that for each $u \in \gamma\inv(0)$, the weight of $s$ in this state is to be considered as $\wt_s(u) = \wt(u)(1-p)^{\beta(u)}$. 
\end{itemize}
{\bf The coverage problem at state $s$}: The PBDS instance at state $s$ is $( \Gi = (\Xi^+, E(\Xi^+), p, \wt_s))$ and budget $b$. $\wt_s$ is defined as follows: For each $u \in \Xi^+$, 
\begin{equation}
\label{eqn:twdpws}
\wt_s(u) = \begin{cases}
\big((1-p)^{\cb(u)}\big)\wt(u) & \text{ if } u \in \Xi \text{ and } c(u) = 0,\\
\wt(u) & \text{ otherwise,}
\end{cases}
\end{equation}
\noindent
In  the following presentation the usage of $\Gi$ and $\mathcal{C}_\bli$ will always be at a specfic state which will be made clear in the context.  
$\soli[s]$ is a subset of $\Xi^+$ which satisfies all the constraints specified by state $s$ and $\vali[s]=\Ci(\Xi^+,\soli[s]S)$ is the maximum value among $\Ci(\Xi^+,S)$ over all $S \subseteq \Xi^+$ and $|S| \leq b$.
In other words it is the optimum solution for Uni-{\sc PBDS} problem on instance $\langle \Gi, b\rangle$ and it satisfies the constraints specified by $s$.
A state $s = (b,\cg,\ca,\cb)$ at a node $\bli$ is said to be \emph{invalid} if there no feasible solution that satisfies the constraints specified by $s$, and $\soli[s] = \texttt{undefined}$ and $\vali[s] = \texttt{undefined}$. If there is a feasible solution, the state is called valid. \\  

  \noindent
{\bf State induced at a node in $\TD$ by a set}:
For a set $D \subseteq V$ of size $k$, we say that $D$ {\em induces} a state $s=(b,\cg, \ca, \cb)$ at node $\bli$
and $s$ is defined as follows:
\begin{itemize}
\item $b = |D \cap \Xi^+|$,
\item The function $\cg:\Xi\to\{0,1\}$ is defined as follows- for each $u\in D \cap \Xi$, $\cg(u) = 1$ and $\cg(u)=0$, for each $u \in D \setminus \Xi$.
\item The functions $\ca, \cb:\Xi\to[0,k]$ are defined as follows-  for each $u\in \cg\inv(0)$, $\ca(u) = |N(u) \cap \Xi^+ \cap D|$, and $\cb(u) = |N(u) \cap (V\setminus \Xi^+) \cap D|$.
\end{itemize}
Depending on $\ca$ and $\cb$ that there can be  different states induced by a set $D$.
\subsubsection{Recursive definition of $\soli$ and $\vali$}
\label{sec:recdef}
For each node $\bli \in V(\TD)$  and $s = (b,\cg,\ca,\cb) \in \Si$, we show how to compute $\soli[s]$ and $\vali[s]$ from the  tables at the children of $\bli$.  $\soli[s]$ and $\vali[s]$ are recursively defined below and we prove a statement on the structure of an optimal solution based on the type of the node $\bli$ in $\TD$.
These statements are used in Section \ref{sec:bottomup} to prove the correctness of the bottom-up evaluation.  
%

\noindent {\bf Leaf node}: Let $\bli$ be a leaf node with bag $\Xi = \emptyset$. The state set $\S_\bli$ is a singleton set with a state $s = (0, \emptyset\to\{0,1\}, \emptyset\to[k], \emptyset \to [k])$. Therefore,
$\soli[s] = \emptyset \text{ and } \vali[s] = 0.$
This can be computed in constant time.
\begin{lemma}
\label{lem:leaf}
The table entries for the state $s$ at a leaf node is computed optimally. 
\end{lemma}
\begin{proof}
The correctness follows from the fact that the graph $\Gi$ is a \emph{null graph}. 
Thus, for a null graph and a valid state $s$, empty set with coverage value zero is the only optimal solution. 
\end{proof}
\noindent {\bf Introduce node}: Let $\bli$ be an introduce node with child $\blj$ such that $\Xi = \Xj \cup \{v\}$ for some $v \notin \Xj$.  
Since $\bli$ is an introduce node, all the neighbors of $v$ in $\G[\Xi^+]$ are in $\Xi$. 
Thus, $N(v) \cap \Xi^+ = N(v) \cap \Xi$.   
In case, $\alpha(v) \not= |N(v) \cap A|$, then, by definition of a solution at a state, the state $s$ does not have a feasible solution. 
Therefore, the state $s$ is invalid. 
Next consider that in state $s$,  $\alpha(v) = |N(v) \cap A|$,
We  define the state $s_\blj$ and define $\soli[s]$ in terms of $\solj[s_\blj]$. 
The state $s_\blj$ differs based on whether $\cg(v) = 0$ and $\cg(v) = 1$. 
For the case  $\cg(v)=0$ the solution to be computed for the state $s$ must not contain $v$, and  for the case $\cg(v)=1$, solution to be computed must contain $v$.  

\noindent
In the case $\cg(v)=0$, define $s_\blj = (b, \cg_\blj, \ca_\blj, \cb_\blj)$ to be the state from  node $\blj$ where
the functions $\cg_\blj:\Xj\to\{0,1\}$, $\alpha_\blj:\Xj\to[0,k]$ and $\beta_\blj:\Xj\to[0,k]$ are as follows:  for each $u \in \Xj$, $\cg_\blj(u) = \cg(u)$, $\alpha_\blj(u) = \alpha(u)$ and $\beta_\blj(u) = \beta(u)$. 
If the state $s_\blj$ is invalid then the state $s$ is also invalid. 
Therefore, we consider that the state $s_\blj$ is valid. 
Then the solution at state $s$ as follows:
\begin{equation}
\label{eqn:dpinsolcv0}
\soli[s] = \solj[s_\blj] 
\end{equation}
and
\begin{equation}
\label{eqn:dpinvalcv0}
\vali[s] = \valj[s_\blj]+\big( 1- (1-p)^{\ca(v)} \big)\wt_s(v)
\end{equation}
\noindent Next, we consider the second case that $\cg(v) = 1$. 
Let $D_v = N(v) \cap \cg\inv(0)$, that is the set of neighbors of $v$ which are not to be selected in $\soli[s]$.
Indeed, these vertices must be considered as their contribution to the expected coverage will increased due to the {\em introduction} of $v$.  
Let $s_\blj = (b-1, \cg_\blj, \ca_\blj, \cb_\blj)$ be the state in the node $\blj$ where
the functions $\cg_\blj:\Xj\to\{0,1\}$, $\alpha_\blj:\Xj\to[0,k]$ and $\beta_\blj:\Xj\to[0,k]$ are defined as follows:  for each $u \in \Xj$, $\cg_\blj(u) = \cg(u)$,
$$\alpha_\blj(u) = \begin{cases} \alpha(u)-1 & \text{ if } u \in D_v,\\
\alpha(u) & \text{ otherwise},\end{cases}$$ 
and $$\beta_\blj(u) = \begin{cases} \beta(u)+1 & \text{ if } u \in D_v,\\
\beta(u) & \text{ otherwise}.\end{cases}$$  
Note the increase and decrease of $\alpha$ and $\beta$ at each $u \in D_v$: this is to take care of the fact that
a neighbor $v$ has been introduced at node $\bli$ and the aim is to compute a solution that contains $v$. Therefore, at $s_\blj$ for each $u \in \cg\inv(0)$ we consider $\alpha_\blj(u) = \alpha(u)-1$ to reflect the fact that $v$ is not in $\Xj$. To ensure that the solution computed at state $s_\blj$ gives us the desired solution at $s$, for each $u \in \cg\inv(0)$ we consider $\beta_\blj(u) = \beta(u)+1$. Recall that this ensures that $\wt_{s_\blj}(u) = (1-p)^{\beta(u)+1}\wt(u) = (1-p)\wt_s(u)$.  This is used in the coverage expressions in the proof of Lemma \ref{lem:intr}.

\noindent
We now define $\soli[s]$ and $\vali[s]$.
If the state $s_\blj$ is invalid then the state $s$ also invalid. 
Therefore, we consider that the state $s_\blj$ is valid. 
The solution for the state $s$ is defined as follows:
\begin{equation}
\label{eqn:dpinsolcv1}
\soli[s] = \solj[s_\blj] \cup \{v\}
\end{equation}
and 
\begin{equation}
\label{eqn:dpinvalcv1}
\vali[s] = \valj[s_\blj] + \sum_{u \in D_v} p(1-p)^{\alpha(u)-1}\wt_s(u) + \wt_s(v)
\end{equation}
In both the cases, the state $s_\blj$ can be computed in $\Oh(w)$ time. 

\begin{lemma}
\label{lem:intr}
Let $\bli$ be an introduce node in $\TD$ and let $s$ and $s_\blj$ be as defined above.  Let $S$ be a solution of optimum coverage at the state  $s$, then $S \setminus \{v\}$ gives the optimum coverage at the state $s_\blj$ in the node $\blj$.
\end{lemma}
\begin{proof}
Since $S$ is solution at state $s$, by definition $S$ induces the state $s$ at node $\bli$.  
In this proof $\Ci(\Xi^+,S)$ is the coverage at state $s$ and $\Cj(\Xj^+,S)$ is the coverage at state $s_\blj$.
The proof technique is to rewrite $\Ci(\Xi^+,S)$ as a sum of $\Cj(\Xj^+,S \setminus \{v\})$ and an additional term that depends only on $v$, its neighborhood and the state $s$.  This shows that $S \setminus \{v\}$ induces the state $s_\blj$ and attains the maximum coverage at $s_\blj$. 
We consider two cases based on whether $v \in A$ or not. 
In the case $v\notin A$, first it follows that $v \notin S$. Further, by definition 
we have $\Ci(\Xi^+, S) = \Ci(\Xi^+\setminus\{v\}, S) + \Ci(v,S)$. 
Since $s_\blj$ is identical to $s$ except for the vertex $v$ which is not in $\Xj$, it follows that for each $u \in \Xj$, $\wt_{s_{\blj}}(u) = \wt_s(u) $. Therefore,  $\Cj(\Xj^+, S) = \Ci(\Xi^+\setminus\{v\}, S) $.  Further, the state induced by the solution $S$ at node $\blj$ is the state $s_\blj$, and in this case
$S \setminus \{v\}$ is $S$ itself. Therefore, the optimum coverage at $s_\blj$ is achieved by $S \setminus \{v\}$. 

\noindent In the case $v\in A$, it follows that $v \in S$. Further, by definition of $s_\blj$, and the fact that $S$ induces the state $s$, it follows that  $S\setminus \{v\}$ induces the state $s_\blj$.
The coverage value $\Ci(\Xi^+, S)$ can be written as follows:
\begin{eqnarray*}
\Ci(\Xi^+, S) &=& \Ci(\Xi^+\setminus \Xi, S) + \Ci(\Xi\setminus(D_v \cup \{v\}), S) + \Ci(D_v, S) + \Ci(v,S)\\
&=& \Cj(\Xj^+\setminus \Xj, S\setminus\{v\}) + \Cj(\Xj\setminus D_v, S\setminus\{v\}) + \Ci(D_v, S) + \wt_s(v)
\end{eqnarray*}
The first equality follows by partitioning $\Xi^+$ into four sets so that the coverage of each set by $S$ summed up gives the coverage of $\Xi^+$ by $S$. 
In the second equality, the first two terms are re-written as coverage by $S \setminus \{v\}$ at node $\blj$, and $\Ci(v,S)$ is re-written as $\wt_s(v)$.  Since $\bli$ is an introduce node,  the key facts used in this equality are that $\Xi^+\setminus \Xi = \Xj^+\setminus \Xj$ and $\Xi \setminus(D_v \cup \{v\} =  \Xj\setminus D_v$. The other key facts used are that
$\cg_\blj$ and $\gamma$ are identical on $\Xj$, $\ca_\blj$ and $\ca$ are identical on $\Xj^+\setminus \D_v$, and $\cb_\blj$ and $\cb$ are identical on $\Xj^+\setminus \D_v$.  
The coverage value of $\Ci(D_v, S)$ is now re-written as follows, using the key fact that for all $u \in \D_v$, $\gamma(u) = 0$, and $|N(u) \cap S| = \alpha(u)$, and that the state induced by solution $S$ at node $\bli$  is the state $s$:
\begin{eqnarray*}
\Ci(D_v, S) &=& \sum_{u \in D_v} \Ci(u,S) = \sum_{u \in D_v} (1-(1-p)^{|N(u) \cap S|})\wt_s(u) = \sum_{u \in D_v} (1-(1-p)^{\ca(u)})\wt_s(u)\\
&=& \sum_{u \in D_v} (1-(1-p)^{\ca(u)})(1-p)^{\cb(u)}\wt(u)\\ 
&=& \sum_{u \in D_v} \Big(\big(1-(1-p)^{\ca(u)-1}\big)(1-p)^{\cb(u)+1} + p(1-p)^{\ca(u)-1}(1-p)^{\cb(u)}\Big)\wt(u)\\
&=& \sum_{u \in D_v} \Big(\big(1-(1-p)^{\ca(u)-1}\big)\wt_{s_\blj}(u) +  \sum_{u \in D_v}p(1-p)^{\ca(u)-1}\wt_s(u)\\
&=& \sum_{u \in D_v} \Cj(u, S\setminus\{v\}) +  \sum_{u \in D_v}p(1-p)^{\ca(u)-1}\wt_s(u)\\
&=& \Cj(D_v, S\setminus\{v\}) + \sum_{u \in D_v}p(1-p)^{\ca(u)-1}\wt_s(u)
\end{eqnarray*}
By substituting the value of $\Ci(D_v, S)$ into $\Ci(\Xi^+,S)$, we  get the following:
\begin{eqnarray*}
\Ci(\Xi^+, S) &=& \Cj(\Xj^+\setminus D_v, S\setminus\{v\}) + \Cj(D_v, S\setminus\{v\}) + \sum_{u \in D_v}p(1-p)^{\ca(u)-1}\wt_s(u) + \wt_s(v)\\
&=& \Cj(\Xj^+, S\setminus\{v\}) + \sum_{u \in D_v}p(1-p)^{\ca(u)-1}\wt_s(u) + \wt_s(v)
\end{eqnarray*}
In this case also, we have rewritten $\Ci(\Xi^+, S)$ as sum of  $\Cj(\Xj^+, S\setminus\{v\})$  and term dependent on $v$ and the state $s$.  
Thus, since $S$ is an optimal solution for state $s$ in $\bli$, it follows that $S \setminus \{v\}$ is an optimal solution for the state $s_\blj$ in $\blj$. 
\end{proof}

\noindent {\bf Forget node}: Let $\bli$ be a forget node with child $\blj$ such that $\Xi = \Xj \setminus \{v\}$ for some $v \in \Xj$.
Since $\bli$ is a forget node,  $N(v) \cap (V\setminus \Xi^+) = \emptyset$, that is, 
all neighbors of $v$ in $\G$ are in $\Gi$. Further, for each vertex $u \in \Xi$, $N[u] \cap \Xi^+ = N[u] \cap \Xj^+$. 
We define the state $s_\blj$ and define $\soli[s]$ in terms of $\solj[s_\blj]$. 
We  consider all possible values of $\cg(v)$, $\ca(v)$, and $\cb(v)$ to define the state $s_\blj$. 
These values specify the different states in $\blj$. 
For each $z \in \{0,1\}$, define $\cg^z:\Xj \to \{0,1\}$ as follows: for each $u \in \Xj$, $$\cg^z(u) = \begin{cases} \cg(u) & \text{ if } u \not= v,\\
z & \text{ u = v }.\end{cases}$$ 
The parameter $z$ specifies  whether $v$ should be in the desired solution or not. 
For each $x \leq [0,b]$, define $\alpha^{x}:\Xj \to [0,k]$ as follows: $$\alpha^{x}(u) = \begin{cases} \alpha(u) & \text{ if } u \not= v,\\
x & \text{ u = v.}\end{cases}$$ 
In the case when $\cg^{z}(v)=0$,  then the parameter $x$ specifies the number of neighbors of $v$ which should be in the desired solution. 
Define $\beta':\Xj \to [0,k]$ as follows: $$\beta'(u) = \begin{cases} \beta(u) & \text{ if } u \not= v,\\
0 & \text{ u = v.}\end{cases}$$ 
For each $z \in \{0,1\}$ and each $x \in [0,b]$, let $s^{z,x}$ denote the state $(b, \cg^z, \alpha^{x}, \beta')$ in $\blj$. \\
If for each $z \in \{0,1\}$ and each $x \in [0,b]$,  $s^{z,x}$ is invalid, then the state $s$ also invalid. 
Therefore, we consider that there exists a 2-tuple $z \in \{0,1\}$ and each $x \in [0,b]$ such that the state $s^{z,x}$ is valid. 
Further, we define the following 2-tuple as follows:
\begin{equation}
\label{eqn:dpfnargmax}
z',x' = \argmax_{\substack{z \in \{0,1\}, x \in [0,b]\\ s^{z,x}\text{ is valid}}} \valj[s^{z,x}].
\end{equation} 
Define $s_\blj = s^{z',x'} = (b, \cg^{z'},\alpha^{x'}, \beta')$ and the solution at the state $s$ as follows:
\begin{equation}
\label{eqn:dpfnsol}
\soli[s] = \solj[s_\blj],
\end{equation}
and
\begin{equation}
\label{eqn:dpfnval}
\vali[s] = \valj[s_\blj].
\end{equation}
The state $s_\blj$ can be computed in $\Oh( k)$ time.
\begin{lemma}
\label{lem:forget}
Let $\bli$ be a forget node in $\TD$, and let $s$ and $s_\blj$ be as defined above. If $S$ is a solution of optimum coverage at state $s$, then $S$ is a solution of optimum coverage at the state $s_\blj$ in node $\blj$. 
\end{lemma}
\begin{proof}
Let $\hat{s}$ be the state induced by $S$ in the node $\blj$. In the following argument, the coverage $\Ci(\Xi^+,S)$ is considered at state $s$ and the coverage $\Cj(\Xj^+,S)$ is considered at state $\hat{s}$ in node $\blj$.
Since $\Xi^+ = \Xj^+$, $\Ci(\Xi^+, S) = \Ci(\Xj^+, S) = \Cj(\Xj^+, S) = \valj[\hat{s}]$. 
Let us define $z = |\{v\} \cap S|$ and $x = |N[z] \cap S|$.  From Equation~\ref{eqn:dpfnargmax}, it is clear that the $z', x'$ chosen is such that $\valj[s^{z',x'}] \geq \valj[s^{z,x}] = \valj[\hat{s}]$. Since $s_\blj$ is defined as $s^{z',x'}$, $\valj[s_\blj] \geq \valj[s^{z,x}] = \valj[\hat{s}] = \Cj(\Xj^+,S)$.  Since $S$ is an optimal solution at state $s$, and since $\Xi^+ = \Xj^+$, it follows that $\Cj(\Xj^+,S) \geq \valj[s_\blj]$.  Therefore $\vali[s] = \Ci(\Xi^+,S) = \Cj(\Xj^+,S) = \valj[s_\blj]$. Hence the lemma.
\end{proof}

\noindent {\bf Join node}: Let $\bli$ be a join node with children $\blj$ and $\blh$ such that $\Xi = \Xj = \Xh$. 
We define the states $s_\blj$ and $s_\blh$, and define $\soli[s]$ in terms of $\solj[s_\blj]$ and $\solh[s_\blh]$.  $s_\blj$ and $s_\blh$ are selected from a set consisting of $O(k^\omega)$ elements.  

\noindent
To define $s_\blj$ and $s_\blh$, we observe that since $\Xi = \Xj = \Xh$, $\cg\inv(1)$ is contained in both $\Xj$ and $\Xh$. Therefore, the $\cg$ gets carried over  from $s$ to $s_\blj$ and $s_\blh$.

\noindent
Next, we identify the candidate values for the budget in the two states $s_\blj$ and $s_\blh$.  We know that in a solution $S$ which induces the state $s$ at node $\bli$ , $|\cg\inv(1)|$ vertices are in $\Xi$ and $b-|\cg\inv(1)|$ vertices are in $\Xi^+\setminus \Xi$. 
Since $\Xi^+\setminus \Xi$ can be partitioned into two sets $\Xj^+\setminus \Xj$ and $\Xh^+\setminus \Xh$, we consider a parameter $z$ to  partition the value $b-|\cg\inv(1)|$. 
For each $0 \leq z \leq b-|\cg\inv(1)|$, let $b_{\blj,z} = |\cg\inv(1)| + z$ and $b_{\blh,z} = b-z$: we consider states at nodes $\blj$ and $\blh$ with { budget} $b_{\blj,z}$ and $b_{\blh,z}$, respectively.  In other words, for each $0 \leq z \leq b-|\cg\inv(1)|$, we search for a solution of size $b_{\blj,z}$ for a subproblem on $\Xj^+$ which contains $\cg\inv(1)$ and $z$ vertices from $\Xj^+ \setminus \Xj$.  Symmetrically, in node $\blh$, we search for a solution of size $b_{\blh,z}$ for a subproblem on  $\Xh^+$ which contains $\cg\inv(1)$ and $b-\cg\inv(1)-z$ vertices from $\Xh^+ \setminus \Xh$.  These two solutions taken together gives a solution fo size $b$ at state $s$ in node $\bli$.  To ensure that the constraints specified by $\ca$ and $\cb$ are met, we next consider appropriate candidate functions in $s_\blj$ and $s_\blh$.\\

\noindent
 The different candidate functions to obtain $\ca_\blj$ and $\ca_\blh$ in the states $s_\blj$ and $s_\blh$, respectively are based on the coverage of vertices in $\cg\inv(0)$. For each $u \in \cg\inv(0)$, $\alpha(u)$ and $\beta(u)$ are  distributed between $s_\blj$ and $s_\blh$, respectively. The number of possible ways in which this can be done is defined by the following set $\A$.  
Let $$ \A = \{\eta:\cg\inv(0) \to [0,b] \mid \text{for each } u \in \cg\inv(0), 0 \leq \eta(u) \leq \alpha(u)-|N(u) \cap \cg\inv(1)|\}. $$ 
A function $\eta \in \A$ specifies that for each $u \in \cg\inv(0)$, $|N(u) \cap \cg\inv(1)|+\eta(u) $  neighbors of $u$  from $\Xj^+$ should be in the solution from the node $\blj$ and $\alpha(u)-\eta(u)$  neighbors of $u$ from $\Xh^+$ should be in the solution from the node $\blh$.  In particular, in the solution at node $\blj$, $\eta(u)$ neighbours of $u$ must be in $\Xj^+ \setminus \Xj$ and, in the solution at node $\blh$, $\alpha(u) - \eta(u) - |N(u) \cap \cg\inv(1)|$ neighbors of $u$ must be in $\Xh^+ \setminus \Xh$.
More precisely, for each $\eta \in \A$, we define the following functions. 
Let $\alpha_{\blj,\eta}:\Xj\to[0,k]$ such that for each $u \in \Xj$,
$$ \alpha_{\blj, \eta}(u) = \begin{cases} \alpha(u) & \text{ if } \cg(u) = 1,\\
|N(u) \cap \cg\inv(1)| + \eta(u) & \text{ if } \cg(u) = 0.\end{cases}$$
Let $\alpha_{\blh,\eta}:\Xh\to[0,k]$ such that for each $u \in \Xh$,
$$ \alpha_{\blh, \eta}(u) = \begin{cases} \alpha(u) & \text{ if } \cg(u) = 1,\\
\alpha(u)-\eta(u) & \text{ if } \cg(u) = 0.\end{cases}$$
Note that $N(u) \cap \cg\inv(1)$ is counted in both the nodes $\blj$ and $\blh$, and we take care of this after  identifying the candidates functions for $\beta_\blj$ and $\beta_\blh$. Recall, that at state $s$ in node $\bli$, for each $u \in \cg\inv(0)$,  $\soli[s]$ provides a coverage value for $u$ as if it has $\alpha(u) + \beta(u)$ neighbors in $\soli[s]$.  Among these, exactly $\alpha(u)$ must be selected from $\Xi^+$.  To ensure that this constraint is met, we consider the following candidate functions for $\beta_{\blj}$ and $\beta_{\blh}$.\\
For each $\eta \in \A$, define $\beta_{\blj,\eta}:\Xj\to[0,k]$ such that for each $u \in \Xj$,
$$ \beta_{\blj, \eta}(u) = \begin{cases} \beta(u) & \text{ if } \cg(u) = 1,\\
\beta(u) + \alpha(u)-|N(u) \cap \cg\inv(1)|-\eta(u) & \text{ if } \cg(u) = 0.\end{cases}$$
Observe that $\alpha_{\blj,\eta}(u)=|N(u) \cap \cg\inv(1)|-\eta(u)$ is subtracted from $\alpha(u)+\beta(u)$. Symmetrically, $\alpha_{\blh,\eta}(u)$ is  subtracted from $\alpha(u)+\beta(u)$ to obtained a candidate $\beta_\blh$.
Let $\beta_{\blh,\eta}:\Xh\to[0,k]$ such that for each $u \in \Xh$,
$$ \beta_{\blh, \eta}(u) = \begin{cases} \beta(u) & \text{ if } \cg(u) = 1,\\
\beta(u) + \eta(u) & \text{ if } \cg(u) = 0.\end{cases}$$ \\

\noindent
Using the candidate values for budget, $\cg$, $\ca$, and $\cb$ at  node $\blj$ and node $\blh$, we now specify the set of candidate states to be considered at nodes $\blj$ and $\blh$.
For each $0 \leq z \leq b-|\cg\inv(1)|$ and each $\eta \in \A$, let $s_{\blj,z,\eta} = (b_{\blj,z}, \cg, \alpha_{\blj,\eta}, \beta_{\blj,\eta})$ and $s_{\blh,z,\eta} = (b_{\blh,z}, \cg, \alpha_{\blh,\eta}, \beta_{\blh,\eta})$. \\

\noindent
Next, to write $\soli[s]$ in terms of coverage of $\Xj^+$ and $\Xh^+$ we need to identify vertices whose contribution the coverage would be over-counted when we take the union of $\soli[s_\blj]$ and $\soli[s_\blh]$. Since $\Xi = \Xj = \Xh$, two cases need to be handled.
First, the vertices on $\cg\inv(1)$ are counted twice, once in $\Xj$ and once in $\Xh$.  Therefore, in the expansion of $\Ci(\Xi^+, S)$ in terms of the coverage at $\Xj$ and $\Xh$, we will have to subtract out $\sum_{u \in \cg\inv(1)}\wt_s(u)$.  Secondly,  for each $\eta \in \A$ and each $u \in \cg\inv(0)$, the coverage of $u$ by $\cg\inv(1)$ is counted twice, once in $\Xj$ and once in $\Xh$.  To subtract this over-counting we introduce the following function $\lambda$ associated with the state $s$ at the join node $\bli$. It is an easy arithmetic exercise to verify that $\lambda(\eta,u)$ is the value that must be subtracted.
For each $\eta \in \A$ and for each $u \in \cg\inv(0)$, let $$ \lambda(\eta,u) = \big( (1-p)^{\alpha(u)-|N(u) \cap \cg\inv(0)|-\eta(u)} + (1-p)^{\eta(u)} - (1-p)^{\alpha(u)}-1\big)\wt_s(u).$$

\noindent
Finally, we come to the recursive specification of $\soli[s]$. 
If for each $0 \leq z \leq b-|\cg\inv(1)|$ and each $\eta \in \A$, either $s_{\blj,z,\eta}$ or $s_{\blh,z,\eta}$ is invalid, then $s$ is also invalid. 
Therefore, we consider those values of $0 \leq z \leq b-|\cg\inv(1)|$ and $\eta \in \A$ such that both the states $s_{\blj,z,\eta}$ and $s_{\blh,z,\eta}$ are valid. 
Further, we define the following tuple :
\begin{equation}
\label{eqn:dpjnargmax}
z',\eta' = \argmax_{\substack{0 \leq z \leq b-|\cg\inv(1)|, \eta \in \A,\\ s_{\blj,z,\eta} \text{ and } s_{\blh,z,\eta} \text{ are valid}}}\valj[s_{\blj,z,\eta}] + \valh[s_{\blh,z,\eta}] - \sum_{u \in \cg\inv(0)}\lambda(\eta,u)
\end{equation}
Define $s_\blj = s_{\blj,z',\eta'} = (b_{\blj,z'},\cg,\alpha_{\blj,\eta'}, \beta_{\blj,\eta'})$ and $s_\blh = s_{\blh,z',\eta'} = (b_{\blh,z'},\cg,\alpha_{\blh,\eta'}, \beta_{\blh,\eta'})$.
Then, the solution for the state $s$ is defined as follows:
\begin{equation}
\label{eqn:dpjnsol}
\soli[s] = \solj[s_\blj] \cup \solh[s_\blh],
\end{equation}
and
\begin{equation}
\label{eqn:dpjnval}
\vali[s] = \valj[s_\blj] + \valh[s_\blh] - \sum_{u \in \cg\inv(0)} \lambda(\eta',u)-\sum_{u \in \cg\inv(1)}\wt_s(u).
\end{equation}
The cardinality of the set $\mathcal{A}$ is at most $\Oh(k^{w})$. This is the most dominant term in the size of the recursive definition. This is becuase for each fixed $\eta \in \mathcal{A}$, $\ca$ and $\cb$ are uniquely defined.  Then, the states $s_\blj$ and $s_\blh$ can be computed in $\Oh(k^{w+1}wn)$ time. 
\begin{lemma}
\label{lem:join}
Let $s$ be a state in the join node $\bli$ in $\TD$ and let $s_\blj$ and $s_\blh$ be as defined above. 
Let $S$ be a solution of optimum coverage at state $s$, then $\Ci(\Xi^+, S) \leq \vali[s]$.
\end{lemma}
\begin{proof}
Let $\hat{s} = (\hat{b}, \hat{\cg}, \hat{\ca}, \hat{\cb})$ and $\tilde{s} =(\tilde{b}, \tilde{\cg}, \tilde{\ca}, \tilde{\cb})$ be the  states induced at the nodes $\blj$ and $\blh$ by the set $S$, respectively.
In this context, we consider the coverage functions $\Ci(\cdot,\cdot)$, $\Cj(\cdot,\cdot)$ and $\Ch(\cdot,\cdot)$ are considered at states $s$, $\hat{s}$ and $\tilde{s}$, respectively. \\
Let $S_\blj = S \cap (\Xj^+ \setminus \Xj)$ and $S_\blh = S \cap (\Xh^+ \setminus \Xh)$.
Let $z = |S_\blj|$ and $\eta:\cg\inv(0)\to[0,b]$ such that for each $u \in \cg\inv(0)$, $\eta(u) = N(u) \cap (\Xj^+\setminus \Xj)$.   We now define $\Ci(\Xi^+, S)$ in terms of $\Cj(\Xj^+, S)$, $\Ch(\Xh^+, S)$, and a subtracted term dependent on $z$, $\eta$, and $s$.  This is done as follows and we ensure that the coverage is exactly counted. 
\begin{eqnarray}
\Ci(\Xi^+, S) &=& \Ci(\Xi^+\setminus\Xi, S) + \Ci(\Xi, S) \nonumber\\
&=& \Ci(\Xj^+\setminus \Xj, S) + \Ci(\Xh^+\setminus \Xh, S) + \Ci(\cg\inv(1), S) + \Ci(\cg\inv(0), S) \nonumber \\
&=& \Cj(\Xj^+\setminus \Xj, S_\blj) + \Ch(\Xh^+\setminus \Xh, S_\blh) + \Ci(\cg\inv(1), S) + \sum_{u \in \cg\inv(0)} \Ci(u, S) \label{eqn:cov-expan}
\end{eqnarray}
The first equality follows by the partition $\Xi^+$ into $\Xi^+\setminus \Xi$ and $\Xi$. 
In the second equality, the set $\Xi^+\setminus \Xi$ is partitioned into $\Xj^+\setminus \Xj$ and $\Xh^+ \setminus \Xh$, and the set $\Xi$ is partitioned in $\cg\inv(0)$ and $\cg\inv(1)$.
The third equality  follows from the fact that for each $u \in \Xj^+\setminus \Xj$, $\wt_{\hat{s}}(u) = \wt_s(u)$ and $N[u] \cap S = N[u] \cap S_\blj$. 
Similarly, for each $u \in \Xh^+\setminus \Xh$, $ \wt_{\hat{s}}(u)= \wt_s(u)$ and $N[u] \cap S = N[u] \cap S_\blh$. \\

\noindent
Next, we consider the term $\sum_{u \in \cg\inv(0)} \Ci(u, S) $.  For each $u \in \cg\inv(0)$, $\Ci(u,S)$ can be written as sum of coverages of $u$ by the sets $S_\blj$ and $S_\blh$. In particular, here we carefully use the values of $\hat{\ca}$ and $\hat{\cb}$ at node $\blj$ and $\tilde{\ca}$ and $\tilde{\cb}$ at node $\blh$. Further, each of the equations follows by simple arithmetic and the definition of $\Cj(u,S_\blj)$, $\Cj(u,S_\blh)$, $\hat{\alpha}$, and $\tilde{\alpha}$.
\begin{eqnarray*}
\Ci(u, S) &=& \big(1-(1-p)^{\alpha(u)}\big)\wt_s(u) \\
&=& \big(1-(1-p)^{\ca(u)}\big)(1-p)^{\cb(u)}\wt(u)\\
&=& \Big(\big(1-(1-p)^{\hat{\ca}(u)}\big)(1-p)^{\hat{\cb}(u)}\wt(u)\Big) + \Big(\big(1-(1-p)^{\tilde{\ca}(u)}\big)(1-p)^{\tilde{\cb}(u)}\wt(u)\Big)\\
&& - \big( (1-p)^{\hat{\ca}(u)-\cg\inv(1)} + (1-p)^{\tilde{\ca}(u)-\cg\inv(1)} - (1-p)^{\ca(u)}-1\big)(1-p)^{\cb(u)}\wt(u)\\
&=& \Cj(u, S_\blj) + \Ch(u, S_\blh)\\
&& - \big( (1-p)^{|N(u) \cap (\Xj^+\setminus \Xj)|} + (1-p)^{|N(u) \cap (\Xh^+\setminus \Xh)|} - (1-p)^{\alpha(u)}-1\big)\wt_s(u)\\
&=& \Cj(u, S_\blj) + \Ch(u, S_\blh)\\
&& - \big( (1-p)^{\alpha(u)-|N(u) \cap \cg\inv(0)|-\eta(u)} + (1-p)^{\eta(u)} - (1-p)^{\alpha(u)}-1\big)\wt_s(u)\\
&=& \Cj(u, S_\blj) + \Ch(u, S_\blh) - \lambda(\eta, u)
\end{eqnarray*}
We substitute this in Equation \ref{eqn:cov-expan} and derive the following equalities.
\begin{eqnarray*}
\Ci(\Xi^+, S)&=& \Cj(\Xj^+\setminus \Xj, S_\blj) + \Ch(\Xh^+ \setminus \Xh, S_\blh) + \Ci(\cg\inv(1), S)\\  
&& + \sum_{u \in \cg\inv(0)} \Big(\Cj(u, S_\blj) + \Ch(u, S_\blh) - \lambda(\eta, u)\Big)\\
&=& \Cj(\Xj^+ \setminus \Xj, S_\blj) + \Ch(\Xh^+ \setminus \Xh, S_\blh) + \Cj(\cg\inv(0), S_\blj) + \Ci(\cg\inv(1), S) \\
&& + \Ch(\cg\inv(0), S_\blh) + \Ci(\cg\inv(1), S)  - \Ci(\cg\inv(1), S) -\sum_{u \in \cg\inv(0)}\lambda(\eta, u)
\end{eqnarray*}
Using $\hat{\alpha}$ and $\hat{\beta}$, it follows that  $\Cj(\Xj, S) = \Cj(\cg\inv(0), S_\blj) + \Ci(\cg\inv(1), S)$. Similarly, using $\tilde{\alpha}$ and $\tilde{\beta}$, it follows that  $\Ch(\Xh, S) = \Ch(\cg\inv(0), S_\blh) + \Ci(\cg\inv(1), S)$. Therefore,
\begin{eqnarray*}
\Ci(\Xi^+, S)&=& \Cj(\Xj^+ \setminus \Xj, S_\blj) + \Ch(\Xh^+ \setminus \Xh, S_\blh) + \Cj(\Xj, S) + \Ch(\Xh, S) \\
&&  - \Ci(\cg\inv(1), S) -\sum_{u \in \cg\inv(0)}\lambda(\eta, u) 
\end{eqnarray*}
Putting together the terms corresponding to $\Cj(\cdot,\cdot)$ and $\Ch(\cdot,\cdot)$ it follows that
\begin{eqnarray*}
\Ci(\Xi^+, S)&=& \Cj(\Xj^+, S) + \Ch(\Xh^+, S) -  \Ci(\cg\inv(1), S) -\sum_{u \in \cg\inv(0)}\lambda(\eta, u) 
\end{eqnarray*}
Since $\hat{s}$ and $\tilde{s}$ are states at nodes $\blj$ and $\blh$ induced by the state $S$, we know that
$\Cj(\Xj^+, S) \leq  \valj[\hat{s}]$ and $\Ch(\Xh^+, S) \leq \valh[\tilde{s}]$.  Therefore, 
\begin{eqnarray*}
\Ci(\Xi^+, S) &\leq& \valj[\hat{s}] + \valh[\tilde{s}]  - \sum_{u \in \cg\inv(1)}\wt_s(u) - \sum_{u \in \cg\inv(0)} \lambda(\eta, u)
\end{eqnarray*}
Further, by the choice of $z',\eta'$ in Equation ~\ref{eqn:dpjnargmax}, we know that $\hat{s}$ and $\tilde{s}$ do not result in a larger value than $s_{\blj}$ and $s_{\blh}$.  Formally,   we know that
\begin{eqnarray*}
\valj[\hat{s}] + \valh[\tilde{s}]  - \sum_{u \in \cg\inv(1)}\wt_s(u) - \sum_{u \in \cg\inv(0)} \lambda(\eta, u) \leq \vali[s]
\end{eqnarray*}
Hence the Lemma.
\end{proof}
\subsubsection{Bottom-Up Evaluation: Correctness of the DP Formulation}
\label{sec:bottomup}
\noindent {\bf Correctness invariant.} For a node $\bli$ and a valid state $s$ at $\bli$, the recursive definition in Section \ref{sec:recdef} ensures that
\[\soli[s] = \max_{\substack{D \subseteq V, |D| = k,\\ D\text{ induces }s}}\Ci(\Xi^+, D\cap \Xi^+)\] 
We prove this invariant by induction on the height of a node in the proof of the following theorem.
\begin{theorem}\label{theorem:fpt-by-wk}
The Uni-PBDS problem can be solved in time $2^{\Oh(w \log k)}n^2$ where $w$ is treewidth of the input graph.
\end{theorem}
\begin{proof}
  We first show that the bottom-up evaluation of the tables in $\TD$ maintains the correctness invariant.\\
{\em \bf Invariant:}
For each node $\bli$ in $\TD$, and for each state $s \in \Si$, the correctness invariant is maintained for $\soli[s]$.
\begin{proof}[Proof of Invariant]
The proof is by induction on the height of a node in $\TD$. Recall, that the height of a node $\bli$ in the rooted tree $\TD$ is the distance to the furthest leaf in the subtree rooted at $\bli$. 
The base case is when $\bli$ is a leaf node $\TD$ and height is 0 and the proof of the claim follows from Lemma ~\ref{lem:leaf}.
Let us assume that the claim is true for all nodes in $\TD$ of height at most $\ell-1 \geq 0$. 
We now prove that if the claim is true for all nodes of height at most $\ell-1$, then it is true for a node of height $\ell$. 
Let $\bli$ be a node of height $\ell \geq 1$. 
Since $\bli$ is not a leaf node,  its children are at height at most $\ell-1$. 
Therefore, by the induction hypothesis, the correctness invariant is maintained at all the chidlren of $\bli$.   
Now, we prove that the correctness invariant is maintained at node $\bli$. Let $s$ be a state in node $\bli$.  
Let $S$ be an optimum solution at state $s$ in node $\bli$. We show that $\vali[s] = \Ci(\Xi^+,S)$. If $\bli$ is an introduce node then from  Lemma~\ref{lem:intr}, we know that the optimum coverage of $\Xj^+$ is achieved by $S \setminus \{v\}$ at state $s_\blj$ at node $\blj$.  Similarly, if $\bli$ is 
a forget node, then from Lemma ~\ref{lem:forget} we know that the optimum coverage of $\Xj^+$ is achieved by $S$ at state $s_\blj$ at node $\blj$. By the induction hypothesis in both these cases $\solj[s_\blj]$ is  the set which achieves optimum coverage, and this proves that  $\soli[s]$ is the optimum value at state $s$.  
 Further, if $\bli$ is a join node, then from the description of the computation at a join node, we know that $\vali[s]$ is recursively defined using $\vali[s_\blj]$ and $\vali[s_\blh]$ for an appropriate $s_\blj$ and $s_\blh$.  By the induction hypothesis, $\vali[s_\blj]$ and $\vali[s_\blh]$ are the optimal values. Therefore, it follows from Lemma ~\ref{lem:join} that $\C(\Xi^+,S) \leq \vali[s]$, and thus $\vali[s]$ is the optimum value.   
 Therefore, it  follows from the induction hypothesis that the solution and value are correctly computed at state $s$ based on the correct values computed at $s_\blj$ and $s_\blh$.   This completes the proof of the invariant.
 \end{proof}
 \noindent
Finally, at the root node $\blr$, the state set $\S_\blr$ is a singleton set with a state $s = (k, \emptyset\to\{0,1\}, \emptyset\to[0,k], \emptyset \to [0,k])$.  By the induction hypothesis, the solution and the value maintained at this state are indeed the set that achieves the optimum coverage and the value of the coverage, respectively.
Finally, a node $\bli \in V(T)$ can have $k(2k+2)^{2w+3}$ states and each of them can be computed in time $\Oh((k+1)^{w+1}wn)$.
Since the nice tree decomposition $(T,X)$ has $\Oh(w\ n)$ nodes, the  tables at the nodes in $\TD$ can be computed in time $\Oh(w^2(2k+2)^{4w+8}n^2) = 2^{\Oh(w \log k)}n^2$.  This completes the proof of the theorem.
\end{proof}


\versionspacy{
\section{Solving $k$-SPM}

The structure of $k$-SPM problem plays a central role in our paper.
By the Exponential time hypothesis and the reduction presented in Section~\ref{section:np-hardness},
it follows that any solution to $k$-SPM 
must require $N^{\Omega(k)}$ time.
It is interesting to explore how efficiently can a $k$-SPM problem be solved. 

Let $c\geq 0$ be the smallest real such that
$2$-SPM problem has an $\widetilde O(N^{c})$ time algorithm. 
We show that in such a case, $k$-SPM can be solved optimally
in $\widetilde O(N^{c\lceil k/2\rceil })$ time.

Let $\langle\A,k\rangle$ be the optimization version of $k$-SPM problem,
where $\A=\{ (x_1,y_1), \ldots, (x_N,y_N)\}$.
(We assume that $k$ is even, as otherwise, we can increment $k,N$ by $1$, and 
add a new element $(x_0,y_0)$ to $\A$, where $x_0 = \sum_{i=1}^{N}x_i$
and $y_0=1$, to ensure that $i_0$ is always part of the solution).
We now employ color-coding technique of Alon, Yuster, and Zwick~\cite{AlonYZ95}.
Repeat the following for $\alpha=\nex^{2k}$ iterations:
Independently assign each element of $[N]$ a uniformly random color (integer)
from $[1,k]$.
Let $U$ be the collection of different colorings obtained
by taking the union over all the iterations.

\hl{Perhaps explain in more detail?}

If $S$ is our desired optimal solution to $k$-SPM, then the probability that 
$S$ is colorful
\hl{For example, define \& motivate ``colorful''}
with respect to at least one coloring in $U$ is at least
$$\textstyle 1- \Big(1- \frac{k!}{k^k}\Big)^{\nex^{2k}} \leq 
1- \Big(1- \frac{\nex k}{\nex^{k}}\Big)^{\nex^ {2k}} \leq 1-\frac{1}{2^{\Theta(k)}}~.$$

For each coloring $\gamma\in U$, 
if $(J_1,\ldots,J_k)$ is a partition of $[N]$ induced by $\gamma$,
then compute
\versiondense{
(i) a list~$\L_\gamma^1$ consisting of pairs
$(\sum_{z=1}^{k/2}x_{i_z}, \prod_{z=1}^{k/2}y_{i_z})$,
for $(i_1,\ldots,i_{k/2})\in (J_1,\ldots,J_{k/2})$;~and
(ii) a list~$\L_\gamma^2$~consisting of pairs
$(\sum_{z=1+k/2}^{k}x_{i_z}, \prod_{z=1+k/2}^{k}y_{i_z})$,
for $(i_{1+k/2},\ldots,i_{k})\in (J_{1+k/2},\ldots,J_{k})$.
}
\versionspacy{
\begin{description}
\item{(i)} a list $\L_\gamma^1$ consisting of pairs
$(\sum_{z=1}^{k/2}x_{i_z}, \prod_{z=1}^{k/2}y_{i_z})$,
for $(i_1,\ldots,i_{k/2})\in (J_1,\ldots,J_{k/2})$, and
\item{(ii)} a list~$\L_\gamma^2$~consisting of pairs
$(\sum_{z=1+k/2}^{k}x_{i_z}, \prod_{z=1+k/2}^{k}y_{i_z})$,
\\
for $(i_{1+k/2},\ldots,i_{k})\in (J_{1+k/2},\ldots,J_{k})$.
\end{description}
}
Note that an optimal solution to the $k$-SPM $\A$
can be derived by solving the following equation, if the optimal solution $S$
of $\A$ is colorful with respect to at least one coloring in $U$:
\begin{equation}~\label{eq:kspm-exact}
\max_{\gamma~\in~U}~\max_{\substack{(a,b)~\in~\L^1_\gamma\\(\bar a,\bar b)~\in~\L^2_\gamma}}~
\Big(a+\bar a- b\bar b\Big)~.
\end{equation}
By Lemma~\ref{lemma:2-spm-conversion},
the above equation is equivalent to solving $\alpha=|U|$ different $2$-SPM
problems (each of size $O(N^{k/2})$).

The time to compute the lists $\L^1_\gamma$ and $\L^2_\gamma$, for a given
$\gamma\in U$, is $\tilde O(k\cdot N^{k/2})$.
By the transformation given in Lemma~\ref{lemma:2-spm-conversion}
and the assumption that
a $2$-SPM of size $m$ is solvable in $\tilde O(m^c)$ time, we get that
the total time required to solve Eq.~(\ref{eq:kspm-exact}) is
$\tilde O(|U|\cdot k \cdot N^{c\lceil k/2\rceil })$.
(Recall that if $k$ is odd, then we transform the problem and increment $k$
by $1$).
This shows that the $k$-SPM has a Monte-Carlo algorithm 
with running time $\tilde O(\nex^{2k}\cdot k \cdot N^{c\lceil k/2\rceil })$,
and success probability
$1-1/2^{\Theta(k)}$,
hence the following theorem is immediate. 

\begin{theorem}
Let $c> 0$ be the smallest real such that
$2$-{\sc SPM} problem has an $\widetilde O(N^{c})$ time algorithm. 
Then, there exists an $\widetilde O\big((dN)^{c\lceil k/2\rceil}\big)$ time Monte-Carlo algorithm
to solve the $k$-{\sc SPM} problem, for some constant $d>0$.
\end{theorem}

%
%
%
%

}

\versionspacy{
\section{Future Research}
Our research raises a number of directions for future study.
\begin{itemize}
\item
It is fundamental to understand on what graph classes and what failure
models the budgeted dominating set problem switches from being easy to hard.
In this paper we have shown that under the model in which edges fail
independently at random, the problem becomes hard already for trees.
\item
In geometric intersection graphs like disk graphs, the presence of an edge
is dependent upon the radius of disks.  For these graphs very good
approximations are known for dominating set problems, and the natural question
is what happens when the radius of a disk is a random variable.
In particular, one concrete problem is as follows:
The input consists of $n$ disks specified by their centres and each disk has a probability distribution on a finite set of possible radii. The goal is to select $k$ disks such that the expected number of intersected disks is maximized.  Variants of this problem can be used to model different practical scenarios, especially in epidemiology.
\item
The exact position of the PBDS problem in trees in the W-hierarchy is still
unknown. In particular the difficulty is in coming up with a natural circuit characterization with  bounded weft for the existence of a set of $k$ vertices that is expected to dominate  a set of weight at least $k$.  So, could it be that this problem is not in the W-hierarchy at all, is it provable based on the kind of primitives required to describe the property?
\end{itemize}
   
}


\paragraphbf{Acknowledgement.}
We thank an anonymous reviewer for pointing us to
\cite{ArnborgLS91}, yielding a shorter proof of the FPT algorithm for Uni-PBDS
parameterized by treewidth and $k$
\versiondense{
\newpage
}

{\small
\bibliographystyle{plain}
\bibliography{references/report,references/report2,references/uncert}
}


\end{document}